\documentclass[10pt,a4paper]{article}
\usepackage[utf8]{inputenc}
\usepackage{amsmath, amsthm, amssymb, amsbsy, bm, mathrsfs}
\usepackage{enumerate}
\usepackage[svgnames]{xcolor}
\usepackage[toc,page]{appendix}
\usepackage[a4paper]{geometry}
\usepackage{tikz-cd}
\newtheorem{thm}{Theorem}[section]

\newtheorem{tvrz}[thm]{Proposition}
\newtheorem{lemma}[thm]{Lemma}

\newtheorem{theorem}[thm]{Theorem}
\theoremstyle{definition}
\newtheorem{definice}[thm]{Definition}
\theoremstyle{remark}
\newtheorem{rem}[thm]{Remark}
\theoremstyle{definition}
\newtheorem{example}[thm]{Example}

\newcommand{\falpha}{\bm{\alpha}}
\newcommand{\ftau}{\bm{\tau}}

\def\R{\mathbb{R}}

\def\dE{\mathbb{E}}

\def\E{\mathcal{E}}
\def\T{\mathcal{T}}
\def\<{\langle}
\def\>{\rangle}

\def\~{\widetilde}
\def\^{\wedge}
\def\ddt{\left. \frac{d}{dt}\right|_{t=0} \hspace{-0.5cm}}

\def\g{\mathfrak{g}}

\def\fai{\mathbf{i}}

\def\ft{\mathbf{t}}
\def\fb{\mathbf{b}}
\def\fj{\mathbf{j}}

\def\d{\mathfrak{d}}

\def\w{\mathfrak{w}}
\def\X{\mathfrak{X}}
\def\fra{\mathfrak{a}}
\def\io{\mathit{i}}

\def\H{\mathcal{H}}

\def\D{\mathcal{D}}

\def\frD{\mathfrak{D}}

\def\btr{\blacktriangleright}
\def\tr{\triangleright}

\def\tC{\tilde{C}}

\def\fA{\mathbf{A}}

\def\fE{\mathbf{E}}
\def\fM{\mathbf{M}}
\def\fg{\mathbf{g}}
\def\fa{\mathbf{a}}
\def\fB{\mathbf{B}}

\def\fR{\mathbf{R}}
\def\fPi{\mathbf{\Pi}}

\def\f1{\mathbf{1}}

\def\fPsi{\mathbf{\Psi}}

\def\fPhi{\mathbf{\Phi}}
\def\fUps{\mathbf{\Upsilon}}

\def\gTM{\mathbb{T}M}
\def\gTG{\mathbb{T}G}
\def\gTGa{\mathbb{T}G^{\ast}}

\def\gTD{\mathbb{T}D}

\def\gm{\mathbf{G}}
\def\RS{\mathcal{R}}
\def\cD{\nabla}

\def\hcD{\widehat{\nabla}}

\def\cDL{\nabla^{LC}}
\def\cDN{\nabla^{0}}
\def\cif{C^{\infty}(M)}

\newcommand{\bma}[4]{\begin{pmatrix} #1 & #2 \\ #3 & #4 \end{pmatrix}}

\newcommand{\vf}[1]{ \mathfrak{X}^{#1}(M)}

\newcommand{\vfP}[1]{ \mathfrak{X}^{#1}(P)}

\newcommand{\Li}[1]{ \mathcal{L}_{#1}}

\DeclareMathOperator{\Ima}{Im}
\DeclareMathOperator{\vol}{vol}
\DeclareMathOperator{\Ad}{Ad}
\DeclareMathOperator{\ad}{ad}
\DeclareMathOperator{\End}{End}

\DeclareMathOperator{\Hom}{Hom}
\DeclareMathOperator{\LC}{LC}

\DeclareMathOperator{\rk}{rk}
\DeclareMathOperator{\Aut}{Aut}
\DeclareMathOperator{\Hor}{Hor}

\DeclareMathOperator{\Ric}{Ric}

\DeclareMathOperator{\Lie}{Lie}

\DeclareMathOperator{\Tr}{Tr}
\DeclareMathOperator{\At}{At}

\DeclareMathOperator{\Div}{div}

\DeclareMathOperator{\fAd}{\pmb{\Ad}}

\begin{document}
\begin{flushright}
\today
\end{flushright}
\vspace{0.7cm}
\begin{center}

\baselineskip=13pt {\Large \bf{Poisson--Lie T-duality of String Effective Actions: A New Approach to the Dilaton Puzzle}\\}
 \vskip0.5cm
 {\it Dedicated to Ziggy Stardust and the Spiders from Mars}  
 \vskip0.7cm
 {\large{ Branislav Jurčo$^{1}$, Jan Vysoký$^{2}$}}\\
 \vskip0.6cm
$^{1}$\textit{Mathematical Institute, Faculty of Mathematics and Physics,
Charles University\\ Prague 186 75, Czech Republic, jurco@karlin.mff.cuni.cz}\\
\vskip0.3cm

$^{2}$\textit{Institute of Mathematics of the Czech Academy of Sciences \\ Žitná 25, Prague 115 67, Czech Republic, vysoky@math.cas.cz}\\
\vskip0.3cm

\end{center}

\begin{abstract}
For a particular class of backgrounds, equations of motion for string sigma models targeted in mutually dual Poisson--Lie groups are equivalent. This phenomenon is called the Poisson--Lie T-duality. On the level of the corresponding string effective actions, the situation becomes more complicated due to the presence of the dilaton field. 

A novel approach to this problem using Levi-Civita connections on Courant algebroids is presented. After the introduction of necessary geometrical tools, formulas for the Poisson--Lie T-dual dilaton fields are derived. This provides a version of Poisson--Lie T-duality for string effective actions. 
\end{abstract}

\section{Introduction}
Poisson--Lie T-duality is an interesting version of non-Abelian duality for string sigma models. It makes use of the rich geometry of the so called Drinfeld's doubles. It was first proposed by Klimčík and Ševera in \cite{1995PhLB..351..455K}, followed by their papers \cite{Klimcik:1995jn, Klimcik:1995dy}. It was then extensively examined on the quantum level, e.g., in \cite{Alekseev:1995ym, Tyurin:1995bu} and many others. In context of the double field theory, it appeared, e.g., in \cite{Hassler:2017yza}. Let us very briefly recall how Poisson--Lie T-duality (henceforth abbreviated as PLT duality) works. 

Drinfeld double is a quadratic Lie group $D$ with two closed subgroups $G,G^{\ast} \subseteq D$,  such that the respective Lie subalgebras $\g,\g^{\ast} \subseteq \d$ are maximally isotropic, mutually complementary subspaces of the quadratic Lie algebra $(\d,\<\cdot,\cdot\>_{\d})$ corresponding to $D$. Equivalently, this corresponds to a choice of a Manin triple $(\d,\g,\g^{\ast})$ or a Lie bialgebra structure $(\g,\delta)$ on $\g$. For a nice exposition, see, e.g., \cite{Kosmann-Schwarzbach1997}.  We recall the corresponding terminology in Subsection \ref{subsec_LiePoisson}. 

 \textbf{A two-dimensional $\sigma$-model} targeted in $G$ is given by an action functional
\begin{equation} \label{eq_SMaction}
S_{\sigma}[l] = \int_{\Sigma} \< h, l^{\ast}(g)\>_{h} \cdot d\vol_{h} + \int_{\Sigma} l^{\ast}(B),
\end{equation}
on the space of smooth maps $l: \Sigma \rightarrow G$, where $(\Sigma,h)$ is a given oriented two-dimensional Lorentzian manifold\footnote{$\Sigma$ is called a \textbf{worldsheet}, and integration is assumed to make sense, e.g. $\Sigma$ is compact.}, $g$ is a metric on $G$ and $B \in \Omega^{2}(G^{\ast})$ is a given $2$-form. One usually assumes that there are local coordinates $(\tau,\sigma)$ on $\Sigma$, such that $h$ is locally the Minkowski metric. It is convenient to introduce new coordinates $z = \sigma + \tau$ and $\bar{z} = \sigma - \tau$. The functional $S_{\sigma}$ can be then rewritten as 
\begin{equation}
S_{\sigma}[l] = \int_{\R^{2}} dz d\bar{z} \; (l^{\ast}\dE)( \partial_{z},\partial_{\bar{z}}) ,
\end{equation}
where\footnote{We use the blackboard bold character $\mathbb{E}$, in order to distinguish it from $E$ usually denoting the vector bundles.}  $\dE = g + B$.  Without going into details, to each $l: \Sigma \rightarrow G$, one can now assign a unique "Noetherian" $1$-form $J \in \Omega^{1}(\Sigma,\g^{\ast})$. Moreover, one assumes that for $l$ solving the equations of motion, $J$ is subject to the Maurer-Cartan equation with respect to the dual Lie algebra on $\g^{\ast}$: 
\begin{equation} \label{eq_JMCequation}
dJ = \frac{1}{2}[J \^ J]_{\g^{\ast}}. 
\end{equation}
It turns out that this can be ensured by imposing a certain restriction on the background tensor $\dE$. In particular, one can show that each such $\dE$ corresponds to the unique maximal positive definite subspace $\E \subseteq \d$ with respect to $\<\cdot,\cdot\>_{\d}$. The corresponding formula is given in detail in Subsection \ref{subsec_PLT}. 

In this construction, the Lie subgroups $G$ and $G^{\ast}$ play an interchangeable role. One can thus construct a completely analogous action functional $\~S_{\sigma}$ targeted in the dual Lie group $G^{\ast}$. One can hope for some relation of the solutions of both $\sigma$-models. This is indeed possible under some technical assumptions \cite{1995PhLB..351..455K}. 

Let us assume that $l: \Sigma \rightarrow G$ solves the equations given by $S_{\sigma}$ for a background $\dE$ corresponding to the subspace $\E \subseteq \d$. As the corresponding $1$-form $J$ satisfies (\ref{eq_JMCequation}), one can find a smooth map $\~h: \Sigma \rightarrow G^{\ast}$, such that $J = \~h^{\ast} \theta_{R'}$, where $\theta_{R'}$ is the right Maurer-Cartan form on the group $G^{\ast}$. We can now form a single map $d: \Sigma \rightarrow D$, for all $(z,\bar{z}) \in \Sigma$, as 
\begin{equation}
d(z,\bar{z}) = l(z,\bar{z}) \cdot \~h(z,\bar{z}), 
\end{equation}
where the multiplication is taken in the group $D$. Under certain conditions (in detail described in Subsection \ref{subsec_LiePoisson}) one can decompose $d$ the other way around:
\begin{equation}
d(z,\bar{z}) = \~l(z,\bar{z}) \cdot h(z,\bar{z}),
\end{equation} 
for all $(z,\bar{z}) \in \Sigma$. Hence, we find a pair of smooth maps $\~l: \Sigma \rightarrow G^{\ast}$ and $h: \Sigma \rightarrow G$. Also, consider the dual tensor $\~\dE \in \T_{2}^{0}(G^{\ast})$ corresponding to the \textit{same} subspace $\E \subseteq \d$.

\textbf{The main statement of PLT duality:} The map $\~l$ constructed above satisfies the equations of motion given by $\~S_{\sigma}$ for the background $\~\dE$. Moreover, the corresponding Noetherian $1$-form is obtained as $\~J = h^{\ast}\theta_{R}$. 

Motivated by the quantization of a bosonic string, one can add a new term to the action (\ref{eq_SMaction}). Namely, let $\phi \in C^{\infty}(G)$ be a smooth function called the \textbf{dilaton field}, and let $\RS^{(2)} \in C^{\infty}(\Sigma)$ be the scalar curvature of the metric $h$. Consider a new action functional as 
\begin{equation}
S'_{\sigma}[l] = S_{\sigma}[l] + \int_{\Sigma} (l^{\ast}\phi) \cdot \RS^{(2)} \cdot d\vol_{h}. 
\end{equation}
The quantization of this $\sigma$-model is anomaly-free only for a certain class of backgrounds $(g,B,\phi)$. In particular, one requires the vanishing of so called \textbf{$1$-loop beta functions}. See, e.g., \cite{polchinski1998string} for details. Equivalently, these conditions are obtained as equations of motion of the \textbf{low-energy string effective action} given by the functional 
\begin{equation} \label{eq_effaction}
S[g,B,\phi] = \int_{M} e^{-2\phi} \{ \RS(g) - \frac{1}{2} \<dB,dB\>_{g} + \| \cD^{g} \phi \|_{g}^{2} \} \cdot d\vol_{g}. 
\end{equation}
For the explicit form of equations of motion of this theory using our notation, see ,e.g., \cite{Vysoky:2017epf}. 

There has been a lot of effort to incorporate the dilaton into the PLT duality. It involves a demanding calculation using path integrals, see \cite{Tyurin:1995bu, VonUnge:2002xjf}. It turns out that the question can be literally puzzling \cite{Hlavaty:2006hu, Hlavaty:2004mr}. In particular, one obtains some restrictions on the choices of Lie bialgebras $(\g,\delta)$. 

In this paper, we approach this in a slightly different way. Observe that the dual backgrounds $\dE$ and $\~\dE$ of the respective $\sigma$-models are obtained from a single piece of an algebraic data, namely a choice of the maximal positive definite subspace $\E \subseteq \d$, and the geometrical interplay of the two subgroups $G$ and $G^{\ast}$ within the Lie group $D$. If we think of the respective dilaton fields $\phi$ and $\~\phi$ in the same way, one can expect that there is a similar procedure giving a pair of "mutually dual" dilatons.

Our starting point is a geometrical framework described in \cite{2015LMaPh.tmp...53S, Severa:2016prq}. It shows that the PLT duality should  naturally be viewed in terms of Courant algebroids and their reductions. In particular, the question of their renormalization can be addressed in this way, see \cite{Severa:2016lwc}. On the other hand, equations of motion given by (\ref{eq_effaction}) can  conveniently be described in terms of Courant algebroid connections and properties of their associated curvature tensors. This was done, using slightly different techniques, independently in \cite{2013arXiv1304.4294G, Garcia-Fernandez:2016azr, Garcia-Fernandez:2016ofz} and in our work \cite{Jurco:2015xra, Jurco:2016emw, Jurco:2015bfs} and \cite{Vysoky:2017epf}. The idea of this paper is to combine the two techniques. In particular, the choice of the subspace $\E \subseteq \d$ leading to the Poisson--Lie T-dual $\sigma$-model backgrounds corresponds to the choice of a so called generalized metric. This suggests that Levi-Civita Courant algebroid connections should be involved too. 

We summarize the main result of this paper in Theorem \ref{thm_themain}. To the best of our belief and knowledge,  one of the most interesting aspects of its proof is that our geometrical framework naturally imposes restrictions on admissible Lie bialgebras obtained previously on the quantum level. We obtain a pair of dilaton fields $\phi$ and $\~\phi$ on $G$ and $G^{\ast}$, respectively, such that the backgrounds $(\dE,\phi)$ solve the equations of motion given by (\ref{eq_effaction}) if and only if $(\~\dE,\~\phi)$ are solutions to their dual analogue. Interestingly, it turns out that to find such solutions, it is necessary and sufficient to solve the system of \textit{algebraic equations} for the subspace $\E \subseteq \d$. 

\textbf{The paper is organized as follows: }

In Section \ref{sec_Maninpairs}, we focus on a comprehensive review of the rich geometrical and algebraic content of so called Manin pairs, that is a quadratic Lie algebra with a maximally isotropic subalgebra of half its dimension. This gives us an interesting reformulation of the Atiyah sequence of $D$ viewed as a principal $G$-bundle. It turns out that there are two canonical algebroid structures on the trivial vector bundle $E = D/G \times \d$. We recall this along with all necessary background. The choice of a complement of $\g$ in $\d$ that forms a Lie subalgebra of $\d$ gives us a Lie bialgebra $(\g,\delta)$, a basic algebraic object for PLT duality. Finally, we conclude with a brief description of geometrical objects on the Lie group $G$ integrating the Lie bialgebra $(\g,\delta)$ and of the induced actions of Lie groups on their respective duals.

Next, in Section \ref{sec_connections} we recall how generalized geometry and Courant algebroids can  conveniently be used to describe the equations of motion of low-energy string effective actions. This involves  Courant algebroid connections compatible with the generalized metric encoding the background fields. We recall the necessary definitions of the generalized metric, Courant algebroid connections and the associated curvature tensors. In particular, we formulate the essential Theorem \ref{thm_eomconn}. 

Section \ref{sec_reduction} describes in detail  the geometry behind the Poisson--Lie T-duality. In particular, the Drinfel'd double $D$ can be viewed as two different principal bundles, one over the point $\{ \star \}$ with the structure group $D$, the other over the left coset space $D/G$ with the structure group $G$. Consequently, we can construct a $D$-equivariant (and hence also a $G$-equivariant) Courant algebroid on the generalized tangent bundle $TD \oplus T^{\ast}D$, which allows for a reduction onto two different Courant algebroids, one over the point $\{ \ast \}$ and the other over the coset space $D/G$. The first one happens to be just the Lie algebra $\d$ itself, whereas the latter one is isomorphic to the canonical exact Courant algebroid on the generalized tangent bundle $\mathbb{T}G^{\ast}$.  We show that one can lift the generalized metric $\E \subseteq \d$ on $\d$ and then use the reduction by $G$ to obtain a subbundle $V_{+}^{\g} \subseteq \gTGa$. It turns out that this one is a graph of the map $\~\dE$ encoding the backgrounds of the dual sigma model $\~S_{\sigma}$. 

Next, in Section \ref{sec_dilatons}, we show that a similar procedure of lifting and subsequently reducing (with respect to the subgroup $G$) can be performed also for any given Levi-Civita connection (with respect to the generalized metric $\E \subseteq \d$) $\cD^{0}$ on $\d$. It provides us with a new Levi-Civita connection $\~\cD$ on the vector bundle $\gTGa$ with respect to the generalized metric $V_{+}^{\g}$. One has to find conditions for $\~\cD$ to be suitable to describe the equations of motion of the low-energy effective action. First, we address the vanishing of some characteristic vector field $X_{\~\cD}$ associated with $\~\cD$. The most difficult part comes with the calculation of the dilaton field. Although the result is relatively simple, one has to use all tricks up his sleeve to arrive to its final form in Theorem \ref{thm_dilaton}. Based on suggestions from a referee, we have included a more detailed discussion of the dilaton formula in Subsection \ref{subsec_analysis}. 

The rather short final Section \ref{sec_PLT} postulates and proves the main statement of this paper, Theorem \ref{thm_themain}. A reader interested only in the results can skip directly to this part. 
It is followed by several remarks discussing the technical assumptions and possible generalizations of the theorem. 
\section{Geometry of Manin pairs} \label{sec_Maninpairs}
In this section, we will briefly recall the canonical algebraic structures induced by the presence of the Lagrangian subgroup $G \subseteq D$ of a quadratic Lie group $D$. For a detailed exposition of this topic, see \cite{2007arXiv0710.0639B}. By a \textbf{Manin pair}, we mean a pair $(\d,\g)$ of Lie algebras, where $\g$ is a maximally isotropic subalgebra of a quadratic Lie algebra $(\d,\<\cdot,\cdot\>_{\d})$, where $\<\cdot,\cdot\>_{\d}$ is assumed to have the signature $(m,m)$. In particular, this implies $\dim{\g} = m$. Furthermore, we assume that $\d = \Lie(D)$ and $\g = \Lie(G)$, where $G \subseteq D$ is a closed connected Lie subgroup of $D$. We can thus consider the principal $G$-bundle $\pi: D \rightarrow M$, where $M = D / G$ is a quotient manifold of left cosets corresponding to the subgroup $G$. 
\subsection{Atiyah sequence corresponding to the Manin pair} \label{subsec_Atiyah}
As for any principal $G$-bundle, one can construct the canonical short exact sequence (in the category of smooth vector bundles over $M$) called the \textbf{Atiyah sequence of $D$}.
\begin{equation}
\begin{tikzcd}
0 \arrow[r] & \g_{D} \arrow[r, hookrightarrow] & \At(D) \arrow[r,twoheadrightarrow] & TM \arrow[r] & 0 
\end{tikzcd},
\end{equation}
where $\At(D) := TD/G$ is the \textbf{Atiyah vector bundle} obtained by the reduction of the $G$-equivariant tangent bundle $TD$. See, e.g., \cite{Mackenzie} for details of the construction. Also, $\g_{D} := D \times_{\Ad} \g$ denotes the adjoint bundle associated to the principal fibration $\pi: D \rightarrow M$. 

For a Manin pair $(\d,\g)$, the Atiyah sequence can be rewritten as follows. Let $\fPsi_{R}: D \times \d \rightarrow TD$ denote the trivialization of the tangent bundle $TD$ using the right-invariant vector fields.  It induces a global trivialization $\fPsi_{R}^{\natural}: M \times \d \rightarrow \At(D)$ of the Atiyah vector bundle, which can be viewed as a vector bundle isomorphism over $1_{M}$. Moreover, there is a canonical vector bundle isomorphism $\fPhi \in \Hom(T^{\ast}M,\g_{D})$ induced by the Manin pair property $\d / \g \cong \g^{\ast}$. One can thus construct a new (dashed) short exact sequence making the diagram
\begin{equation} \label{eq_AtiyahManin}
\begin{tikzcd}
0 \arrow[r] \arrow[rd, dashed] & \g_{D} \arrow[r, hookrightarrow] & \At(D) \arrow[r,twoheadrightarrow] & TM \arrow[r] & 0 \\
& T^{\ast}M \arrow[u,"\fPhi"'] \arrow[r, "\rho^{\ast}", hookrightarrow, dashed] & M \times \d \arrow[ur,"\rho"', dashed] \arrow[u,"\fPsi_{R}^{\natural}"'] & 
\end{tikzcd} 
\end{equation}
commutative. The maps $\rho$ and $\rho^{\ast}$ can be described explicitly. Indeed, let $\tr$ denote the canonical left action of $D$ on the coset space $M = D / G$, that is $d \tr \pi(d') := \pi(dd')$ for all $d,d' \in D$. Then
\begin{equation} \label{eq_anchoronM}
\rho(m,x) = (\#^{\tr}(x))_{m}, 
\end{equation}
for all $m \in M$ and $x \in \d$, where $\#^{\tr}: \d \rightarrow \vf{}$ denotes the infinitesimal generator corresponding to the action $\tr$ of $D$ on $M$. Now, observe that $E \equiv M \times \d$ comes equipped with a fiber-wise metric $\<\cdot,\cdot\>_{\d}$ induced by the one (denoted by the same symbol) on the quadratic Lie algebra $\d$. Let $g_{\d} \in \Hom(E,E^{\ast})$ be the induced vector bundle isomorphism. Then one obtains 
\begin{equation}
\rho^{\ast} = g_{\d}^{-1} \circ \rho^{T},
\end{equation}
where $\rho^{T} \in \Hom(T^{\ast}M,E^{\ast})$ denotes the (fiber-wise) transpose of the vector bundle map $\rho$. Note that $\Ima(\rho^{\ast}) \subseteq E$ forms a Lagrangian subbundle of $(E, \<\cdot,\cdot\>_{\d})$.
\subsection{Lie and Courant algebroids on $E$} \label{subsec_LieCourant}
It is a well-known fact that there is a canonical Lie algebroid structure on $\At(D)$ induced by the usual Lie bracket on $\X(D)$ known as \textbf{Atiyah Lie algebroid}. Let us recall some definitions first.
\begin{definice}
Let $q: E \rightarrow M$ be a smooth vector bundle, equipped with a vector bundle map $\rho \in \Hom(E,TM)$ called the \textbf{anchor}, together with an $\R$-bilinear bracket $[\cdot,\cdot]_{E}: \Gamma(E) \times \Gamma(E) \rightarrow \Gamma(E)$, subject to the following axioms:
\begin{enumerate}
\item There holds a \textbf{Leibniz rule}: $[\psi, f \psi']_{E} = f [\psi,\psi']_{E} + (\rho(\psi).f) \psi'$.
\item $(\Gamma(E),[\cdot,\cdot]_{E})$ forms a Leibniz algebra, that is, it satisfies the \textbf{Leibniz identity}:
\begin{equation} \label{eq_LI}
[\psi,[\psi',\psi'']_{E}]_{E} = [[\psi,\psi']_{E},\psi'']_{E} + [\psi', [\psi,\psi'']_{E}]_{E}.
\end{equation}
\end{enumerate}
All conditions are supposed to hold for all $f \in \cif$ and $\psi,\psi',\psi'' \in \Gamma(E)$. Then $(E,\rho,[\cdot,\cdot]_{E})$ is called a \textbf{Leibniz algebroid}\footnote{Such a structure is also known under the name Loday algebroid.}. Whenever $[\cdot,\cdot]_{E}$ is skew-symmetric, it is called a \textbf{Lie algebroid}. In this case the Leibniz identity becomes a usual Jacobi identity and $(E,[\cdot,\cdot]_{E})$ forms a Lie algebra. 
\end{definice}
\begin{rem}
For any Leibniz algebroid $(E,\rho,[\cdot,\cdot]_{E})$, $\rho$ becomes a bracket homomorphism:
\begin{equation} \label{eq_rhohom}
\rho([\psi,\psi']_{E}) = [\rho(\psi),\rho(\psi')],
\end{equation}
for all $\psi,\psi' \in \Gamma(E)$. This is forced by the compatibility of the Leibniz rule and (\ref{eq_LI}). 
\end{rem}
The Atiyah Lie algebroid $(\At(D), \widehat{T}(\pi), [\cdot,\cdot]_{\At})$ is induced by the canonical $\cif$-module isomorphism $\Gamma(\At(D)) \cong \X_{G}(D)$, where $\X_{G}(D)$ is a $\cif$-module of $G$-invariant vector fields on $D$ involutive under the usual Lie bracket $[\cdot,\cdot]$ on $\X(D)$. The anchor map $\widehat{T}(\pi) \in \Hom(\At(D),TM)$ is naturally induced by the tangent map to the projection $\pi: D \rightarrow M$. 

Using the vector bundle isomorphism $\fPsi_{R}^{\natural}$, we obtain an isomorphic Lie algebroid $(E,\rho,[\cdot,\cdot]'_{\At})$. We can explicitly describe it as follows. Let $[\cdot,\cdot]_{\d}: \Gamma(E) \times \Gamma(E) \rightarrow \Gamma(E)$ denote the $\cif$-bilinear fiber-wise extension of the Lie algebra bracket $[\cdot,\cdot]_{\d}$. Moreover, as $\Gamma(E) \cong C^{\infty}(M,\d)$, we can act by vector fields $\rho(\psi)$ on sections in $\Gamma(E)$. The bracket $[\cdot,\cdot]'_{\At}$ then takes the form
\begin{equation}
[\psi,\psi']'_{\At} = \rho(\psi).\psi' - \rho(\psi').\psi - [\psi,\psi']_{\d},
\end{equation}
for all $\psi,\psi' \in \Gamma(E)$. Note that the unusual minus sign comes from the fact that we use the right trivialization $\fPsi_{R}^{\natural}$ in the process. It is easy to verify all axioms of Lie algebroid explicitly. Now, observe that the natural generalization of the ad-invariance condition
\begin{equation}
\rho(\psi).\<\psi',\psi''\>_{\d} = \< [\psi,\psi']'_{\At}, \psi'' \>_{\d} + \< \psi', [\psi,\psi'']'_{\At} \>_{\d}
\end{equation}
does not hold for all sections\footnote{One has to restrict $\psi',\psi''$ to be the sections of the kernel subbundle of $\rho$, in which case $(E,\rho,[\cdot,\cdot]'_{\At})$ forms a prototypical example of the so called quadratic transitive Lie algebroid.} of $E$. However, there is an another algebroid structure on $E$ at our disposal. Let us recall its definition. 
\begin{definice} \label{def_Courant}
Let $(E,\rho,[\cdot,\cdot]_{E})$ be a Leibniz algebroid, equipped with a fiber-wise metric $\<\cdot,\cdot\>_{E}$. Suppose that:
\begin{enumerate}
\item Pairing and bracket are compatible: 
\begin{equation}
\rho(\psi).\<\psi',\psi''\>_{E} = \< [\psi,\psi']_{E}, \psi''\>_{E} + \< \psi', [\psi,\psi'']_{E} \>_{E},
\end{equation}
for all $\psi,\psi',\psi'' \in \Gamma(E)$. 
\item Let $g_{E} \in \Hom(E,E^{\ast})$ be the vector bundle isomorphism induced by $\<\cdot,\cdot\>_{E}$ and define an $\R$-linear map $\D: \cif \rightarrow \Gamma(E)$ as the composition $\D := \rho^{\ast} \circ d \equiv g_{E}^{-1} \circ \rho^{T} \circ d$. Then
\begin{equation} \label{eq_CAsympart}
[\psi,\psi']_{E} + [\psi',\psi]_{E} = \D \<\psi,\psi'\>_{E},
\end{equation}
for all $\psi,\psi' \in \Gamma(E)$, i.e., the symmetric part of $[\cdot,\cdot]_{E}$ is controlled by $\rho$ and $\<\cdot,\cdot\>_{E}$.
\end{enumerate}
Then $(E,\rho,\<\cdot,\cdot\>_{E},[\cdot,\cdot]_{E})$ is called a \textbf{Courant algebroid}.
\end{definice}
Courant algebroids originally appeared in \cite{liu1997manin} as a generalization of Manin triples for Lie bialgebroids. The modern definition was introduced in \cite{1999math.....10078R}. 
\begin{rem}
For any Courant algebroid, one has $\rho \circ \rho^{\ast} = 0$. We thus obtain a cochain complex
\begin{equation} \label{eq_CASES}
\begin{tikzcd}
0 \arrow[r] & T^{\ast}M \arrow[r,"\rho^{\ast}"] & E \arrow[r,"\rho"] & TM \arrow[r] & 0
\end{tikzcd}.
\end{equation}
If this is a short exact sequence (i.e. the corresponding cohomology group is trivial), $E$ is called an \textbf{exact Courant algebroid}.
\end{rem}
It turns out that we can find such a structure on $E = M \times \d$ with the fiber-wise metric $\<\cdot,\cdot\>_{\d}$. One can define a new bracket $[\cdot,\cdot]_{E}$ as 
\begin{equation}
\< [\psi,\psi']_{E}, \psi''\>_{\d} := \< [\psi,\psi']'_{\At}, \psi'' \>_{\d} + \< \rho(\psi'').\psi,\psi'\>_{\d},
\end{equation}
for all $\psi,\psi',\psi'' \in \Gamma(E)$. Let $\{ t_{\alpha} \}_{\alpha =1}^{\dim{\d}}$ be some basis of the Lie algebra $\d$, and let $\{ \psi_{\alpha} \}_{\alpha=1}^{\dim{\d}}$ be the corresponding constant sections of $\Gamma(E)$. Write $\psi = \varphi^{\alpha} \cdot\psi_{\alpha}$ for a collection of unique smooth functions $\varphi^{\alpha} \in \cif$. We can then rewrite $[\cdot,\cdot]_{E}$ as 
\begin{equation} \label{eq_Ebracket}
[\psi,\psi']_{E} = [\psi,\psi']'_{\At} + \< \psi_{\alpha}, \psi' \>_{\d} \cdot \D{\varphi^{\alpha}}.
\end{equation}
Exactness of (\ref{eq_AtiyahManin}) implies $\rho \circ \D = 0$. Consequently, as $[\cdot,\cdot]'_{\At}$ is a Lie algebroid bracket, $\rho$ is a bracket homomorphism (\ref{eq_rhohom}) with respect to $[\cdot,\cdot]_{E}$. Moreover, the above expression is $\cif$-linear in $\psi'$ and the Leibniz rule for $[\cdot,\cdot]'_{\At}$ thus implies the one for $[\cdot,\cdot]_{E}$. This is enough for the Leibniz identity and metric compatibility to be verified on constant sections. This is trivial and it reduces to the axioms for the quadratic Lie algebra $(\d, \<\cdot,\cdot\>_{\d},[\cdot,\cdot]_{\d})$. Finally, one has
\begin{equation}
\< [\psi,\psi]_{E}, \psi' \>_{\d} = \< \rho(\psi').\psi,\psi\>_{\d} = \frac{1}{2} \rho(\psi').\<\psi,\psi\>_{\d} = \frac{1}{2} \< \D\<\psi,\psi\>_{\d}, \psi''\>_{\d}.
\end{equation}
The polarization of this relation gives the axiom (\ref{eq_CAsympart}). We conclude that $(E, \rho, \<\cdot,\cdot\>_{\d}, [\cdot,\cdot]_{E})$ forms an exact\footnote{The sequence (\ref{eq_AtiyahManin}) is a priori exact.} Courant algebroid. This example first appeared in unpublished works \cite{alekseevxu, Severa:2017oew} and more recently in \cite{2007arXiv0710.0639B, 2015LMaPh.tmp...53S}. As for any exact Courant algebroid, there is a unique de Rham cohomology class $[H]_{dR} \in H^{3}_{dR}(M)$ characterizing $E$, called the \textbf{Ševera class} of $E$. For any representative $H \in \Omega^{3}_{cl}(M)$ of this class, the Courant algebroid $(E,\<\cdot,\cdot\>_{\d},\rho,[\cdot,\cdot]_{E})$ is isomorphic to the \textbf{generalized tangent bundle} $\gTM \equiv TM \oplus T^{\ast}M$ with an $H$-twisted Dorfman bracket $[\cdot,\cdot]_{D}^{H}$, where
\begin{equation} \label{eq_HDorfman}
[(X,\xi),(Y,\eta)]_{D}^{H} = ( [X,Y], \Li{X}\eta - \io_{Y}d\xi - H(X,Y,\cdot)),
\end{equation}
for all $(X,\xi), (Y,\eta) \in \Gamma(\gTM)$. The anchor is given by the projection onto $TM$ and the fiber-wise metric $\<\cdot,\cdot\>_{\mathbb{T}}$ corresponds to the canonical pairing of $1$-forms with vector fields. The closed $3$-form $H$ is constructed as follows. Consider an isotropic\footnote{Such splitting always exists.} splitting $s \in \Hom(TM,E)$ of the short exact sequence (\ref{eq_CASES}), that is, $\rho \circ s = 1_{TM}$ and $\< s(X),s(Y)\>_{\d} = 0$ for all $X,Y \in \vf{}$. Then set 
\begin{equation}
H(X,Y,Z) = - \< [s(X),s(Y)]_{E}, s(Z) \>_{\d},
\end{equation}
for all $X,Y,Z \in \vf{}$. For $E = M \times \d$, we can view the splitting $s \in \Hom(TM,E)$ as a $\d$-valued $1$-form $s \in \Omega^{1}(M,\d)$. As (\ref{eq_CASES}) is the Atiyah sequence in disguise, each isotropic splitting corresponds to the principal bundle connection $A \in \Omega^{1}(D,\g)$, where the horizontal distribution $\Hor_{A}(D) \subseteq TD$ is isotropic with respect to the bi-invariant metric on $D$ induced by $\<\cdot,\cdot\>_{\d}$. One can calculate $H$ directly in terms of the $1$-form $s$, see \cite{2007arXiv0710.0639B} for the proof. One finds $H = - \frac{1}{2} \tC_{3}(s)$, where 
\begin{equation} \label{eq_tC3s}
\tC_{3}(s) := \< f \^ s \>_{\d} + \frac{1}{3!} \< [s \^ s ]_{\d} \^ s \>_{\d}, \; \; f = ds - \frac{1}{2}[s \^ s]_{\d}
\end{equation}
is the Chern-Simons-like $3$-form defined using $s \in \Omega^{1}(M,\d)$ and the bracket $-[\cdot,\cdot]_{\d}$. 
\subsection{Lie bialgebras, Poisson--Lie groups} \label{subsec_LiePoisson}
Every Manin pair $(\d,\g)$ induces the following short exact sequence:
\begin{equation} \label{eq_ManinpairSES}
\begin{tikzcd}
0 \arrow[r] & \g \arrow[r,"i",hookrightarrow] & \d \arrow[r,"\io^{\ast}", twoheadrightarrow] & \g^{\ast} \arrow[r] & 0
\end{tikzcd},
\end{equation}
where $\io^{\ast}$ is the map defined as $\< \io^{\ast}(x), y \> := \<x,\io(y)\>_{\d}$ for all $x \in \d$ and $y \in \g$. Let $j \in \Hom(\g^{\ast},\d)$ be an isotropic splitting of this sequence, that is $\io^{\ast} \circ j = 1_{\g^{\ast}}$ and $\< j(\xi), j(\eta) \>_{\d} = 0$ for all $\xi,\eta \in \g^{\ast}$. The triple $(\d,\g,j)$ is called a \textbf{split Manin pair}. 

Assume that the maximally isotropic subspace $\g' \equiv j(\g^{\ast})$ forms a Lie subalgebra of $\d$. In this case, the split Manin pair is called a \textbf{Manin triple} $(\d,\g,\g')$. Equivalently, this corresponds to the definition of a \textbf{Lie bialgebra} $(\g,\delta)$, where $\delta: \g \rightarrow \g \otimes \g$ is a $1$-cocycle in the Chevalley--Eilenberg cochain complex. Moreover, let $[\cdot,\cdot]_{\g^{\ast}}$ denote the Lie bracket induced on the dual space $\g^{\ast}$. Both objects are related to the split Manin pair $(\d, \g, j)$ via the equation
\begin{equation}
\delta(x)(\xi,\eta) = \<x, [\xi,\eta]_{\g^{\ast}} \> = \< i(x), [j(\xi),j(\eta)]_{\d} \>,
\end{equation}
for all $x \in \g$ and $\xi,\eta \in \g^{\ast}$. Using the splitting $j$, one can identify the Lie algebra $\d$ with the double $\g \oplus \g^{\ast}$. The bracket $[\cdot,\cdot]_{\d}$ can be under this identification written as 
\begin{equation} \label{eq_bracketdouble}
[(x,\xi),(y,\eta)]_{\d} = ( [x,y]_{\g} + \ad^{\ast}_{\xi}y - \ad^{\ast}_{\eta}x, [\xi,\eta]_{\g^{\ast}} + \ad_{x}\eta - \ad_{y}\xi ), 
\end{equation}
for all $(x,\xi), (y,\eta) \in \g \oplus \g^{\ast}$. Here $\ad^{\ast}_{\xi}$ denotes the coadjoint action of Lie algebra $(\g^{\ast},[\cdot,\cdot]_{\g^{\ast}})$ on the dual space $\g$. Observe that the bracket (\ref{eq_bracketdouble}) is obviously symmetric with respect to the interchange $\g \leftrightarrow \g^{\ast}$. The symmetric non-degenerate invariant form $\<\cdot,\cdot\>_{\d}$ takes the form of the canonical pairing on $\g \oplus \g^{\ast}$, that is
\begin{equation}
\< (x,\xi), (y,\eta) \>_{\d} = \eta(x) + \xi(y), 
\end{equation}
for all $(x,\xi), (y,\eta) \in \g \oplus \g^{\ast}$. In the following, we will always identify $\d$ with the induced structures on the direct sum $\g \oplus \g^{\ast}$. For a more detailed exposition of Lie bialgebras see, e.g., \cite{Kosmann-Schwarzbach1997}. 

If $\g$ integrates to a connected Lie group $G$, each Lie bialgebra structure $(\g,\delta)$ corresponds to a unique Poisson--Lie group bivector $\Pi \in \X^{2}(G)$. Let us recall the definition.
\begin{definice}
Let $G$ be a Lie group equipped with a Poisson bivector $\Pi \in \X^{2}(G)$. We say that the tensor field $\Pi$ is \textbf{multiplicative} if $\Pi_{gh} = L_{g \ast}( \Pi_{h}) + R_{h \ast}( \Pi_{g})$ for all $g,h \in G$. Having this property, the pair $(G,\Pi)$ is called a \textbf{Poisson--Lie group}.
\end{definice}

Note that one immediately obtains $\Pi_{e} = 0$. This excludes constant bivectors and symplectic manifolds as non-trivial examples. 
The cobracket $\delta$ is obtained from $\Pi$ via the formula
\begin{equation} \label{eq_deltaPirel}
\delta(x) = (\Li{\~x} \Pi)_{e},
\end{equation}
where $\~x \in \X(G)$ is any vector field extending the tangent vector $x \in \g \equiv T_{e}G$. For details of this construction and detailed exposition of multiplicative tensor fields, see \cite{lu1990, Kosmann-Schwarzbach1997}. 

As $\g^{\ast} \subseteq \d$ is a Lie subalgebra, we may assume that $\g^{\ast} = \Lie(G^{\ast})$ for a connected closed Lie subgroup $G^{\ast} \subseteq D$. To any Lie bialgebra $(\g,\delta)$, there is a dual Lie bialgebra $(\g^{\ast},\delta^{\ast})$. One thus gets a dual Poisson--Lie group structure $\Pi^{\ast} \in \X^{2}(G^{\ast})$. We can use it to define an infinitesimal left action of $\g$ on the Lie group $G^{\ast}$. To any $x \in \g$, we may assign a right-invariant $1$-form $x^{R'} \in \Omega^{1}(G^{\ast})$ where the prime is used to indicate that the right translation is on $G^{\ast}$. Set 
\begin{equation} \label{eq_leftdress}
\varphi^{l}(x) := \Pi^{\ast}(\cdot,x^{R'}) =: \Pi^{\ast}(x^{R'}),
\end{equation}
for all $x \in \g$. Then $\varphi^{l}: \g \rightarrow \X(G^{\ast})$ is an anti-homomorphism of Lie brackets, that is 
\begin{equation} \label{eq_leftdressahom}
\varphi^{l}([x,y]_{\g}) + [\varphi^{l}(x), \varphi^{l}(y)] = 0,
\end{equation}
for all $x,y \in \g$. This follows from the vanishing Schouten bracket $[\Pi^{\ast},\Pi^{\ast}]_{S} = 0$ and the fact that $\Pi^{\ast}$ encodes the bracket $[\cdot,\cdot]_{\g}$ as $(\delta^{\ast})^{T}(x,y) = [x,y]_{\g}$. The map $\varphi^{l}$ is called \textbf{the left infinitesimal dressing action} of $\g$ on $G^{\ast}$. Similarly, one can obtain the right infinitesimal dressing action as $\varphi^{r}(x) = - \Pi^{\ast}(x^{L'})$ for all $x \in \g$. In this case, $\varphi^{r}$ is a bracket homomorphism. 

We can combine the embeddings of both subgroups $G$ and $G^{\ast}$ with the group multiplication on $D$ to obtain a local diffeomorphism from $G \times G^{\ast}$ to $D$. We define it as $(g,h) \mapsto g h$, for all $(g,h) \in G \times G^{\ast}$. Similarly, we have a local diffeomorphism $(h,g) \mapsto h g$ of the Cartesian product $G^{\ast} \times G$ and $D$. 

It can be shown that these two maps are \textit{global} diffeomorphisms iff  either of the infinitesimal actions $\varphi^{l}$ or $\varphi^{r}$ is complete. Assume that this is true. For each $(g,h) \in G \times G^{\ast}$, we can then define $\Phi^{l}: G \times G^{\ast} \rightarrow G^{\ast}$ and $\~\Phi^{r}: G \times G^{\ast} \rightarrow G$ via the relation
\begin{equation}
g h = \Phi^{l}(g,h) \cdot \~\Phi^{r}(g,h),
\end{equation}
for all $(g,h) \in G \times G^{\ast}$. It is easy to see that $\Phi^{l}$ defines a left action of $G$ on $G^{\ast}$, whereas $\~\Phi^{r}$ is a right action of $G^{\ast}$ on $G$. Moreover, $\Phi^{l}$ is precisely the left Lie group action integrating the left infinitesimal dressing action $\varphi^{l}$ of $\g$ on $G$, and $\~\Phi^{r}$ is the right Lie group action integrating the right infinitesimal dressing action $\~\varphi^{r}$ of $\g^{\ast}$ on $G$. For details and proofs, see \cite{luthesis}. $\Phi^{l}$ is called the \textbf{left dressing action} of $G$ on $G^{\ast}$ and similar names are adopted for the other variants. 

One can use this to identify the coset space $M = D / G$ with the Lie group $G^{\ast}$. Indeed, define the smooth map $\fUps: M \rightarrow G^{\ast}$ making the diagram in the category of smooth manifolds
\begin{equation}
\begin{tikzcd}
G \times G^{\ast} \arrow[r] \arrow[d,"\Phi^{l}"] & D \arrow[d,"\pi"] \\
G^{\ast}  & M \arrow[l,"\fUps"', dashed]
\end{tikzcd}
\end{equation}
commutative. One can use the diffeomorphism $\fUps$ to define a Courant algebroid structure on $E' = G^{\ast} \times \d$. Recall that the anchor $\rho$ of the Courant algebroid on $E = M \times \d$ is given by the infinitesimal generator of the $D$-action $\tr$ on $M$, see (\ref{eq_anchoronM}).  Let $\btr$ be the left $D$-action on $G^{\ast}$ making $\fUps$ into a $D$-equivariant map. One can easily read out the action of the subgroups $G$ and $G^{\ast}$ on $G^{\ast}$. For any $g \in G$ and all $h,h' \in G^{\ast}$, one obtains 
\begin{equation} \label{eq_GbtrGast}
g \btr h = \Phi^{l}(g,h), \; \; \; h \btr h' = h \cdot h'.
\end{equation}
In other words, $G$ acts by the left dressing action, whereas $G^{\ast}$ acts by the ordinary left translation. Now, consider the vector bundle map $\rho' \in \Hom(E',TG^{\ast})$ defined by 
\begin{equation}
\begin{tikzcd}
E := M \times \d \arrow[r,"\rho"] \arrow[d,"\fUps \times 1_{\d}"] & TM \arrow[d,"T(\fUps)"] \\
E' := G^{\ast} \times \d \arrow[r,"\rho'", dashed] & TG^{\ast}
\end{tikzcd}.
\end{equation}
It follows that $\rho'( h, (x,\xi)) = (\#^{\btr}(x,\xi))_{h}$, for all $(x,\xi) \in \d \equiv \g \oplus \g^{\ast}$ and $h \in G^{\ast}$. Here, $\#^{\btr}: \d \rightarrow \X(G^{\ast})$ denotes the infinitesimal generator of the action $\btr$. Using the observation (\ref{eq_GbtrGast}), one finds
\begin{equation} \label{eq_rho'anchor}
\rho'(h, (x,\xi)) = (\xi_{R'} + \Pi^{\ast}(x^{R'}))_{h},
\end{equation}
for all $(x,\xi) \in \d$ and $h \in G^{\ast}$. The pairing on $E'$ remains the fiber-wise extension of $\<\cdot,\cdot\>_{\d}$. As constant sections of $E$ are mapped by $\fUps \times 1_{\d}$ onto constant sections of $E'$, the bracket $[\cdot,\cdot]_{E'}$ is again just a fiber-wise extension of the Lie bracket $-[\cdot,\cdot]_{\d}$ using the Leibniz rule. In other words, for all $\psi,\psi',\psi'' \in \Gamma(E')$, we find
\begin{equation}
\< [\psi,\psi']_{E'}, \psi''\>_{\d} = \< \rho'(\psi).\psi' - \rho'(\psi').\psi - [\psi,\psi']_{\d}, \psi'' \>_{\d} + \< \rho'(\psi'').\psi, \psi' \>_{\d}.
\end{equation}
It follows that $(E', \rho', \<\cdot,\cdot\>_{\d}, [\cdot,\cdot]_{E'})$ defines an exact Courant algebroid.  Its Ševera class can easily be calculated as there is a particularly simple isotropic splitting $s' \in \Hom(TG^{\ast},E')$. Let $\psi_{(x,\xi)} \in \Gamma(E')$ denote the constant section corresponding to $(x,\xi) \in \d$. For each $\xi \in \g^{\ast}$ define 
\begin{equation} \label{eq_s'splitting}
s'( \xi_{R'}) := \psi_{(0,\xi)}.
\end{equation}
We have $\rho'(\psi_{(x,\xi)}) = \xi_{R'} + \Pi^{\ast}(x^{R'})$. Plugging $s'$ into the formula (\ref{eq_tC3s}) gives $\tC_{3}(s') = 0$. The cubic term in $s'$ is zero as $\g^{\ast} \subseteq \d$ is an isotropic subalgebra and $f = 0$ as $s'$ is in fact a right Maurer-Cartan form $\theta_{R'}$ viewed as a $\d$-valued $1$-form on $G^{\ast}$. We conclude that the Ševera class $[H']_{dR} \in H^{3}_{dR}(G^{\ast})$ corresponding to the Courant algebroid $(E',\rho',\<\cdot,\cdot\>_{\d}, [\cdot,\cdot]_{E'})$ is trivial. In fact, we have found the splitting $s'$ where $H' = 0$. 
\section{Connections and low-energy effective actions} \label{sec_connections}
Let us very briefly recall how the equations of motion for low-energy string effective actions can be conveniently encoded in terms of Courant algebroid connections. For a detailed discussion of this topic, see \cite{Jurco:2016emw, Vysoky:2017epf}. 
\begin{definice}
Let $(E,\rho,\<\cdot,\cdot\>_{E},[\cdot,\cdot]_{E})$ be a Courant algebroid. We say that an $\R$-bilinear map $\cD: \Gamma(E) \times \Gamma(E) \rightarrow \Gamma(E)$ is a \textbf{Courant algebroid connection} if
\begin{equation} \label{eq_CAconnaxioms}
\begin{split}
\cD_{f \psi} \psi' = & \ f (\cD_{\psi} \psi'), \\
\cD_{\psi}(f \psi') = & \ f \cD_{\psi}\psi' + (\rho(\psi).f) \psi', \\
\rho(\psi).\<\psi',\psi''\>_{E} = & \ \< \cD_{\psi}\psi', \psi''\>_{E} + \< \psi', \cD_{\psi}\psi''\>_{E},
\end{split}
\end{equation}
for all $\psi,\psi',\psi'' \in \Gamma(E)$ and $f \in \cif$. We write $\cD_{\psi} \equiv \cD(\psi,\cdot)$. 

One says that $\cD$ is \textbf{torsion-free} if the $3$-form $T_{\cD} \in \Omega^{3}(E)$ defined by
\begin{equation}
T_{\cD}(\psi,\psi',\psi'') := \< \cD_{\psi}\psi' - \cD_{\psi'}\psi - [\psi,\psi']_{E}, \psi''\>_{E} + \< \cD_{\psi''}\psi,\psi'\>_{E},
\end{equation}
for all $\psi,\psi',\psi'' \in \Gamma(E)$, vanishes identically. 
\end{definice}
The definition of the torsion $3$-form appeared independently in \cite{alekseevxu, 2007arXiv0710.2719G}. The unusual additional term appearing in $T_{\cD}$ is necessary to establish its tensoriality. Let us now recall a generalization of Riemann and Ricci curvature tensors. They were first introduced and examined  in a double field theory paper \cite{Hohm:2012mf}. 
\begin{definice}
Let $\cD$ be a Courant algebroid connection. First, define a map $R^{(0)}_{\cD}$ as 
\begin{equation}
R^{(0)}_{\cD}(\phi',\phi,\psi,\psi') := \< ([\cD_{\psi},\cD_{\psi'}] - \cD_{[\psi,\psi']_{E}} ) \phi, \phi'\>_{E},
\end{equation}
for all $\psi,\psi',\phi,\phi' \in \Gamma(E)$. This is not a tensor on $E$. Instead, set
\begin{equation}
R_{\cD}(\phi',\phi,\psi,\psi') := \frac{1}{2} \{ R_{\cD}^{(0)}(\phi',\phi,\psi,\psi') + R_{\cD}^{(0)}(\psi',\psi,\phi,\phi') + \< \cD_{\psi_{\lambda}} \psi, \psi'\>_{E} \cdot \< \cD_{\psi^{\lambda}_{E}} \phi, \phi'\>_{E} \},
\end{equation}
where $\{ \psi_{\lambda} \}_{\lambda=1}^{\rk{E}}$ is some local frame for $E$, $\{ \psi^{\lambda} \}_{\lambda=1}^{\rk{E}}$ is the corresponding dual frame and $\psi^{\lambda}_{E} \equiv g_{E}^{-1}(\psi^{\lambda})$. It follows that $R_{\cD} \in \T_{4}^{0}(E)$ and we call it the \textbf{generalized Riemann curvature tensor} of $\cD$. It has only one non-trivial (up to a sign) contraction in two indices, namely 
\begin{equation}
\Ric_{\cD}(\psi,\psi') := R_{\cD}( \psi^{\lambda}_{E}, \psi', \psi_{\lambda}, \psi),
\end{equation}
for all $\psi,\psi' \in \Gamma(E)$. We call it the \textbf{generalized Ricci curvature tensor}. Finally, set 
\begin{equation}
\RS_{\cD}^{E} := \Ric_{\cD}( \psi^{\lambda}_{E}, \psi_{\lambda}) 
\end{equation} 
to define the \textbf{Courant--Ricci scalar curvature}. Clearly $\RS^{E}_{\cD} \in C^{\infty}(M)$. 
\end{definice}
For the sake of brevity, we sometimes drop some of the adjectives and say for example just \textit{curvature tensor}. $R_{\cD}$ enjoys symmetries similar to the usual curvature tensor in Riemannian geometry. In particular, the Ricci tensor $\Ric_{\cD}$ is symmetric. Note that all objects in the previous definition are well-defined for general Courant algebroid connections. This is in contrast with the approach taken in \cite{2013arXiv1304.4294G, Garcia-Fernandez:2016ofz} where the presence of additional structures on $E$ is necessary. 

\begin{definice}
Let $(E,\<\cdot,\cdot\>_{E})$ be an orthogonal vector bundle.\footnote{That is equipped with a fiber-wise metric $\<\cdot,\cdot\>_{E}$.} The choice of a maximal positive definite subbundle $V_{+} \subseteq E$ with respect to $\<\cdot,\cdot\>_{E}$ is called a \textbf{generalized Riemannian metric}. It gives an orthogonal decomposition $E = V_{+} \oplus V_{-}$, where $V_{-} = V_{+}^{\perp}$ is a maximal negative definite subbundle of $E$. Equivalently, $V_{\pm}$ form $\pm 1$ eigenbundles to the orthogonal involution $\tau \in \Aut(E)$, such that
\begin{equation}
\gm(\psi,\psi') := \< \psi, \tau(\psi') \>_{E},
\end{equation}
for all $\psi,\psi' \in \Gamma(E)$, defines a positive definite fiber-wise metric $\gm$ on $E$. 
\end{definice}
Note that, in principle, one can  use a more general definition, where the restriction of $\<\cdot,\cdot\>_{E}$ onto $V_{+}$ defines a fiber-wise metric of an arbitrary signature. The corresponding orthogonal involution $\tau$ is then an arbitrary one. This is used for example in \cite{Garcia-Fernandez:2016ofz}. However, for our purposes this is an unnecessary generality.\footnote{The set of such generalized metrics is pretty big, for example $V_{+} = E$ is then also a generalized metric.} 
\begin{example} \label{ex_genmetricgTM}
Let $(E,\<\cdot,\cdot\>_{E}) = (\gTM, \<\cdot,\cdot\>_{\mathbb{T}})$. Then every generalized metric $V_{+} \subseteq E$ is of the form $\Gamma(V_{+}) = \{ (X,(g+B)(X)) \; | \; X \in \vf{} \}$, where $g > 0$ is a unique Riemannian metric on $M$ and $B \in \Omega^{2}(M)$ is a unique $2$-form on $M$. The corresponding fiber-wise metric $\gm$ can be, with respect to the splitting $\gTM = TM \oplus T^{\ast}M$, written as a formal $2 \times 2$ matrix 
\begin{equation} \label{eq_genmetricusual}
\gm = \bma{g - Bg^{-1}B}{Bg^{-1}}{-g^{-1}B}{g^{-1}}.
\end{equation}
\end{example}
\begin{definice}
Let $\cD$ be a Courant algebroid connection. Let $V_{+} \subseteq E$ be a generalized metric. We say that $\cD$ is \textbf{metric compatible with $V_{+}$} if, for each $\psi \in \Gamma(E)$, one has $\cD_{\psi}(\Gamma(V_{+})) \subseteq \Gamma(V_{+})$. Equivalently, in terms of the fiber-wise metric $\gm$, we obtain the usual metric compatibility 
\begin{equation}
\rho(\psi). \gm(\psi',\psi'') = \gm( \cD_{\psi}\psi' , \psi'') + \gm( \psi' ,\cD_{\psi}\psi''),
\end{equation}
for all $\psi,\psi',\psi'' \in \Gamma(E)$. We say that $\cD$ is a \textbf{Levi-Civita connection} on $E$ with respect to $V_{+}$ if $\cD$ is metric compatible with $V_{+}$ and it is torsion-free if $T_{\cD} = 0$. We write $\cD \in \LC(E,V_{+}).$
\end{definice}
\begin{example} \label{ex_quadrconn}
Let $(\d,\<\cdot,\cdot\>_{\d},[\cdot,\cdot]_{\d})$ be a quadratic Lie algebra. We can use it to define a Courant algebroid $(\d,0,\<\cdot,\cdot\>_{\d},-[\cdot,\cdot]_{\d})$. Let $\E_{+} \subseteq \d$ be a maximal positive definite subspace of $\d$ defining a generalized metric on $(\d,\<\cdot,\cdot\>_{\d})$. Let $\E_{-} = \E_{+}^{\perp}$ be the corresponding orthogonal complement. For $x \in \d$, let $x_{\pm}$ denote its projections onto $\E_{\pm}$. Set
\begin{equation} \label{eq_cD0formula}
\begin{split}
\< \cD^{0}_{x}y, z \>_{\d} := & \ - \< [x_{+},y_{-}]_{\d}, z_{-} \>_{\d} - \< [x_{-}, y_{+}]_{\d}, z_{+} \>_{\d} \\
& - \frac{1}{3} \< [x_{+},y_{+}]_{\d}, z_{+} \>_{\d} - \frac{1}{3}\< [x_{-},y_{-}]_{\d}, z_{-} \>_{\d},
\end{split}
\end{equation}
for all $x,y,z \in \d$. Then $\cD^{0} \in \LC(\d,\E_{+})$. The minus sign in the Courant algebroid bracket will be explained in Section \ref{sec_reduction}. 
\end{example}
For a given generalized metric $V_{+}$, there exist infinitely many Levi-Civita connections\footnote{Except for some low rank cases.}. See e.g. \cite{Garcia-Fernandez:2016ofz} or \cite{Jurco:2016emw} for details. The presence of $V_{+}$ allows us to introduce two new properties of $\cD$. 

\begin{definice}
Let $\cD$ be a Courant algebroid connection. Let $V_{+} \subseteq E$ be a generalized metric. We say that $\cD$ is \textbf{Ricci compatible with $V_{+}$} if the corresponding Ricci tensor $\Ric_{\cD}$ is block diagonal with respect to the orthogonal decomposition $E = V_{+} \oplus V_{-}$, that is,
\begin{equation}
\Ric_{\cD}(V_{+},V_{-}) = 0.
\end{equation}
Moreover, we can introduce the \textbf{scalar curvature with respect to $\gm$} as 
\begin{equation} \RS_{\cD}^{\gm} := \Ric_{\cD}( \gm^{-1}(\psi^{\lambda}), \psi_{\lambda}). \end{equation}
\end{definice}

To each Courant algebroid connection $\cD$, there is a canonically assigned vector field $X_{\cD} \in \vf{}$ which we call the \textbf{characteristic vector field} of $\cD$. First, one can define a divergence operator $\Div_{\cD}: \Gamma(E) \rightarrow \cif$. For all $\psi \in \Gamma(E)$, put  
\begin{equation}
\Div_{\cD}(\psi) := \< \cD_{\psi_{\lambda}} \psi, \psi^{\lambda} \>,
\end{equation}
where $\{ \psi_{\lambda} \}_{\lambda=1}^{\rk{E}}$ is an arbitrary local frame for $E$. The action of $X_{\cD}$ on $f \in \cif$ is defined by
\begin{equation} \label{eq_charverfield}
X_{\cD}.f := \Div_{\cD}( \D{f}),
\end{equation}
where $\D: \cif \rightarrow \Gamma(E)$ was introduced in Definition \ref{def_Courant}. One has to use both the defining properties of $\cD$ and the Courant algebroid axioms to prove that the action of $X_{\cD}$ satisfies the Leibniz rule and thus defines a vector field on $M$. 

For $E = \gTM$, Levi-Civita connections can  conveniently be used to describe the equations of motion of string low-energy effective actions. We state the result in the form of a theorem, for its proof see \cite{Jurco:2016emw}. 

\begin{theorem} \label{thm_eomconn}
Let $(\gTM, \rho, \<\cdot,\cdot\>_{\mathbb{T}}, [\cdot,\cdot]_{D}^{H})$ be the Courant algebroid structure on the generalized tangent bundle with the $H$-twisted Dorfman bracket (\ref{eq_HDorfman}) for a given closed $3$-form $H \in \Omega^{3}_{cl}(M)$. Let $V_{+} \subseteq \gTM$ be a generalized metric corresponding to $(g,B)$ as in Example \ref{ex_genmetricgTM}. 

Let $\cD \in \LC(\gTM,V_{+})$ be a Levi-Civita connection on $\gTM$ with respect to $V_{+}$, such that its characteristic vector field vanishes, $X_{\cD} = 0$. Moreover, assume that there is a scalar function $\phi \in \cif$, such that 
\begin{equation} \label{eq_dilatonrel}
\< \cD_{\rho^{\ast}( h^{-1}_{\gm}(e_{k}))} \rho^{\ast}(e^{k}) , \rho^{\ast}( h_{\gm}^{-1}(Z)) \>_{E} = \< d\phi, Z \>,
\end{equation}
for all $Z \in \vf{}$. Here $\{ e_{k} \}_{k=1}^{\dim{M}}$ is an arbitrary local frame on $M$, $\rho^{\ast}$ is the map in the short exact sequence (\ref{eq_CASES}) and $h_{\gm}$ is a fiber-wise metric on $T^{\ast}M$ defined by
\begin{equation}
h_{\gm}(\xi,\eta) := \gm( \rho^{\ast}(\xi), \rho^{\ast}(\eta)), 
\end{equation}
for all $\xi,\eta \in \Omega^{1}(M)$. Note that this connection exists for any given $V_{+}$ and $\phi$. 

Then the fields $(g,B,\phi)$ satisfy the equations of motion of the \textbf{low-energy string effective action} given by the functional
\begin{equation}
S[g,B,\phi] = \int_{M} e^{-2\phi} \{ \RS(g) - \frac{1}{2} \<H + dB,H + dB\>_{g} + 4 \| \cD^{g} \phi \| _{g}^{2} \} \cdot d \vol_{g},
\end{equation}
if and only if $\cD$ is Ricci compatible with $V_{+}$ and its scalar curvature $\RS_{\cD}^{\gm}$ is zero. $\RS(g)$ is the usual scalar curvature corresponding to the metric $g$ and $\cD^{g}: \cif \rightarrow \vf{}$ is the gradient operator. 
\end{theorem}
There is a note in order. We have written the condition (\ref{eq_dilatonrel}) in a little bit overcomplicated way. The reason for this is   explained in detail in \cite{Jurco:2016emw}. In a nutshell, the left-hand side of (\ref{eq_dilatonrel}) is invariant under Courant algebroid isomorphisms, that is if we replace $\cD$ with its isomorphic image, and the same with $\rho$ and fiber-wise metrics $\gm$ and $\<\cdot,\cdot\>_{E}$, it does not change. This will be very important in Section \ref{sec_dilatons}. Note that in this case $h_{\gm} = g^{-1}$. 
\section{Reductions and Poisson--Lie T-duality} \label{sec_reduction}
In this section, we will review how the Poisson--Lie T-duality of sigma models fits into the Courant algebroid setting. This is based mostly on the paper \cite{2015LMaPh.tmp...53S}. 
\subsection{Two reductions of the same Courant algebroid} \label{subsec_tworeductions}
First, let us review the general procedure. For a detailed theory of  the reduction of Courant algebroids, see \cite{Bursztyn2007726, Baraglia:2013wua}.
\begin{definice}
Let $(E,\rho,\<\cdot,\cdot\>_{E},[\cdot,\cdot]_{E})$ be a Courant algebroid, where $q: E \rightarrow P$ is a vector bundle over a principal $G$-bundle $\pi: P \rightarrow M$. Let $\g = \Lie(G)$. One says that an $\R$-linear map $\Re: \g \rightarrow \Gamma(E)$ is an \textbf{action of $\g$ on Courant algebroid $E$} if it satisfies
\begin{equation} \label{eq_Courantaction}
\rho \circ \Re = \#, \; \; \Re([x,y]_{\g}) = [\Re(x), \Re(y)]_{E}, 
\end{equation}
for all $x,y \in \g$, where $\#: \g \rightarrow \vfP{}$ is the infinitesimal generator of the principal bundle $G$-action. Moreover, assume that the infinitesimal action $x \gtrdot \psi = [\Re(x),\psi]_{E}$ integrates to a right Lie group action $\fR: E \times G \rightarrow E$ making $E$ into a $G$-equivariant vector bundle. Then $(E,\rho,\<\cdot,\cdot\>_{E},[\cdot,\cdot]_{E})$ is called a \textbf{$G$-equivariant Courant algebroid}.
\end{definice}

The definition ensures that for each $g \in G$, the pair $(\fR_{g},R_{g})$ forms a Courant algebroid automorphism. Having an action of $G$ on the Courant algebroid $E$, there exists a procedure very similar to the symplectic reduction. This allows one to find a reduced Courant algebroid structure on a certain vector bundle $E'_{\g}$ over the base manifold $M$ of the principal bundle $P$. 

First, $\Re$ can be viewed as a vector bundle map from the trivial vector bundle\footnote{The conditions (\ref{eq_Courantaction}) then imply that $\Re: P \times \g \rightarrow E$ can be viewed as a morphism of the action Lie algebroid corresponding to the Lie algebra action $\#: \g \rightarrow \vfP{}$ and of the Courant algebroid $E$. See \cite{Mackenzie}.}  $P \times \g$ to $E$. As the right action of $G$ on $P$ is free, $\Re$ is fiber-wise injective. One can thus define its image subbundle $K_{\g} := \Ima(\Re) \subseteq E$. 

Next, consider the corresponding orthogonal complement $K^{\perp}_{\g}$. Both these subbundles are $G$-invariant with respect to $\fR$ and moreover, the subbundle $K^{\perp}_{\g}$ is involutive under the bracket $[\cdot,\cdot]_{E}$. Finally, we have to take out the  elements isotropic with respect $\<\cdot,\cdot\>_{E}$. One thus sets
\begin{equation}
E'_{\g} := \frac{K^{\perp}_{\g} / G}{(K_{\g} \cap K^{\perp}_{\g}) / G}.
\end{equation}
On the level of sections, we have $\Gamma(E'_{\g}) = \Gamma_{G}(K^{\perp}_{\g}) / \Gamma_{G}( K_{\g} \cap K^{\perp}_{\g})$. The anchor $\rho'_{\g}: \Gamma(E'_{\g}) \rightarrow \vf{}$ is defined via the commutative diagram 
\begin{equation} \label{eq_rho'gdef}
\begin{tikzcd}
& E \arrow[r,"\rho"] \arrow[d,"\natural"] & TP \arrow[d,"T(\pi)"] \\
K^{\perp}_{\g} / G \arrow[d,"\chi"] \arrow[r, hookrightarrow] & E/G \arrow[r, "\hat{\rho}"] & TM \\
E'_{\g} \arrow[rru, "\rho'_{\g}"', dashed, bend right=10] & & 
\end{tikzcd},
\end{equation}
where $\natural$ is the quotient map induced by $\fR$ and $\hat{\rho}$ is the map defined by the equivariant vector bundle morphism $\rho$. Finally, $\chi$ denotes the natural map corresponding to the quotient by the subbundle $(K_{\g} \cap K_{\g}^{\perp}) / G$. The pairing $\<\cdot,\cdot\>_{E'_{\g}}$ and the bracket $[\cdot,\cdot]_{E'_{\g}}$ are obtained in the obvious way from the corresponding structures on $E$. With some effort, one can prove that $(E'_{\g}, \rho'_{\g}, \<\cdot,\cdot\>_{E'_{\g}}, [\cdot,\cdot]_{E'_{\g}})$ is indeed a Courant algebroid, see \cite{Bursztyn2007726} for details.  

Now, let us return back to the scenario described in Section \ref{sec_Maninpairs}. Let $(\d,\g)$ be a Manin pair and let $(D,G)$ be the corresponding pair of Lie groups. Let $E = \gTD$ be the generalized tangent bundle over $D$. Let $H \in \Omega^{3}_{cl}(D)$ be defined as 
\begin{equation} \label{eq_specialH}
H := - \frac{1}{12} \< [\theta_{L} \^ \theta_{L}]_{\d} \^ \theta_{L} \>_{\d},
\end{equation}
where $\theta_{L} \in \Omega^{1}(D,\d)$ denotes the left Maurer-Cartan form. In other words, $H$ is proportional to the canonical bi-invariant $3$-form on the quadratic Lie group $D$. We assume that $\gTD$ is equipped with the Courant algebroid of the $H$-twisted Dorfman bracket. We will now show that there are two different group actions making $\gTD$ into an equivariant Courant algebroid. 

\subsubsection*{Reduction using the whole group $D$} \label{subsec_Dreduction}
First, we can view $D$ as a principal $D$-bundle $\pi_{D}: D \rightarrow \{ \ast \}$ over the singleton manifold $\{ \ast \}$. In this case, the map $\#: \d \rightarrow \X(D)$ assigns to each $x \in \d$ the corresponding left-invariant vector field $\#{x} = x^{L}$. Define $\Re: \d \rightarrow \Gamma(\gTD)$ as 
\begin{equation} \label{eq_Reexplicit}
\Re(x) := (x^{L}, -\frac{1}{2} \< \theta_{L}, x\>_{\d}).
\end{equation}
The induced action $\gtrdot$ of $\d$ on $\Gamma(\gTD)$ is easily calculated to take the explicit form 
\begin{equation}
x \gtrdot (Y,\eta) = ( [x^{L}, Y], \Li{x^{L}} \eta),
\end{equation}
for all $x \in \d$ and $(Y,\eta) \in \Gamma(\gTD)$. This action clearly integrates to the canonical right translation $\fR$ induced by the group multiplication on $D$, whence the map $\Re$ defines the action on $\gTD$ making it into a $D$-equivariant Courant algebroid. 

To find the reduced Courant algebroid, it is convenient to identify $\gTD$ with the trivial vector bundle $E_{\d} = D \times (\d \oplus \d^{\ast})$. As $\Gamma(E_{\d}) = C^{\infty}(D,\d) \oplus C^{\infty}(D,\d^{\ast})$, one can write every section of $E_{\d}$ as a pair $(\Phi,\Psi)$ of vector space valued functions. The anchor $\rho_{\d} \in \Hom(E_{\d},TD)$ can be written as $\rho_{\d}(\Phi,\Psi)_{d} = (\Phi(d))^{R}_{d}$. The $H$-twisted Dorfman bracket with the $3$-form given by (\ref{eq_specialH}) can, in terms of the trivial vector bundle $E_{\d}$, be recast as 
\begin{equation} \label{eq_HDorfmantrivial}
\begin{split}
[(\Phi,\Psi), (\Phi',\Psi')]_{E} = ( & \frD_{\Phi}\Phi' - \frD_{\Phi'}\Phi - [\Phi,\Phi']_{\d}, \frD_{\Phi}\Psi' - \frD_{\Phi'}\Psi - \ad^{\ast}_{\Phi}\Psi' + \ad^{\ast}_{\Phi'}\Psi \\
& + \< \frD \Phi, \Psi'\> + \< \frD \Psi, \Phi'\> + \frac{1}{2} \< [\Phi,\Phi']_{\d}, \cdot \>_{\d} ),
\end{split}
\end{equation}
for all $(\Phi,\Psi), (\Phi',\Psi') \in \Gamma(E_{\d})$. Let $\{ t_{\alpha} \}_{\alpha=1}^{\dim{\d}}$ be any fixed basis of $\d$. We define $\frD_{\Phi}\Phi' := \Phi^{\alpha} \cdot (t_{\alpha}^{R}.\Phi')$ where $\Phi = \Phi^{\alpha} t_{\alpha}$. By $[\cdot,\cdot]_{\d}$ and $\ad^{\ast}$ we mean the corresponding fiber-wise operations in Lie algebra $\d$ and its double $\d^{\ast}$. The canonical pairing $\<\cdot,\cdot\>_{\mathbb{T}}$ on $\gTD$ becomes the canonical pairing $\<\cdot,\cdot\>_{E_{\d}}$ on $E_{\d}$: 
\begin{equation}
\< (\Phi,\Psi), (\Phi',\Psi') \>_{E_{\d}} = \<\Psi',\Phi\> + \<\Psi, \Phi'\>,
\end{equation}
where $\< \Psi',\Phi\>$ is a fiber-wise contraction of two functions valued in mutually dual vector spaces.   The map $\Re$ defined by (\ref{eq_Reexplicit}) viewed as a vector bundle map $\Re \in \Hom(D\times\d,E_{\d})$ takes the form 
\begin{equation}
\Re(d,x) = \big(d, ( \Ad_{d}(x), - \frac{1}{2} \< \cdot, \Ad_{d}(x) \>_{\d} )\big).
\end{equation}
$D$-invariant sections of $E$ correspond to the constant sections, which we denote as $(\Phi_{x}, \Psi_{\eta})$, for $(x,\eta) \in \d \oplus \d^{\ast}$. It is now easy to find the $C^{\infty}( \{ \ast \})$-modules\footnote{Vector spaces.} $\Gamma_{D}(K_{\d})$ and $\Gamma_{D}(K^{\perp}_{\d})$:
\begin{equation} \label{eq_KdKdperpexpl}
\Gamma_{D}(K_{\d}) = \{ (\Phi_{x}, \Psi_{-\frac{1}{2} g_{\d}(x)}) \; | \; x \in \d \}, \; \; \Gamma_{D}(K^{\perp}_{\d}) = \{ (\Phi_{x}, \Psi_{\frac{1}{2} g_{\d}(x)}) \; | \; x \in \d \},
\end{equation}
where $g_{\d} \in \Hom(\d,\d^{\ast})$ is the isomorphism induced by the invariant form $\<\cdot,\cdot\>_{\d}$. As this is non-degenerate, we immediately see that $\Gamma_{D}(K_{\d} \cap K^{\perp}_{\d}) = 0$. We thus get $\Gamma(E'_{\d}) = \Gamma_{D}(K^{\perp}_{\d})$. This yields the obvious isomorphism $E'_{\d} \cong \d$. The anchor $\rho'_{\d}$ is trivially zero. For the pairing, we have
\begin{equation} \label{eq_Apairingonconstant}
\< (\Phi_{x}, \Psi_{\frac{1}{2} g_{\d}(x)}), (\Phi_{y}, \Psi_{\frac{1}{2} g_{\d}(y)}) \>_{E} = \<x,y\>_{\d},
\end{equation}
for all $x,y \in \d$. The pairing induced on $\d$ thus coincides with the original metric $\<\cdot,\cdot\>_{\d}$. For the Courant algebroid bracket, plugging in the constant sections into (\ref{eq_HDorfmantrivial}) gives 
\begin{equation} \label{eq_Ebracketonconstant}
[( \Phi_{x}, \Psi_{\frac{1}{2} g_{\d}(x)}), ( \Phi_{y}, \Psi_{\frac{1}{2} g_{\d}(y)}) ]_{E} = ( \Phi_{-[x,y]_{\d}}, \Psi_{ -\frac{1}{2} g_{\d}([x,y]_{\d})} ),
\end{equation}
for all $x,y \in \d$. But this shows that the induced bracket on $\d$ is $-[\cdot,\cdot]_{\d}$. We conclude that the reduced Courant algebroid $E'_{\d}$ can be identified with the Courant algebroid $(\d,0,\<\cdot,\cdot\>_{\d}, - [\cdot,\cdot]_{\d})$. This also explains the peculiar minus sign we have used in Example \ref{ex_quadrconn}. 

\subsubsection*{Reduction using the Lagrangian subgroup $G$}
Now, we consider the principal $G$-fibration $\pi: D \rightarrow M$, where $M = D / G$ is the left coset space corresponding to the subgroup $G$. Let $\Re: \g \rightarrow \Gamma(\gTD)$ be the action obtained by the restriction of (\ref{eq_Reexplicit}) onto the Lie subalgebra $\g$. First, note that for all $x,y \in \g$, one has 
\begin{equation}
\< \Re(x), \Re(y) \>_{\mathbb{T}} = - \<x,y\>_{\d} = 0, 
\end{equation}
as $\g \subseteq \d$ is assumed to be isotropic. Hence $K_{\g} \subseteq \gTD$ is an isotropic subbundle and thus $K_{\g} \subseteq K_{\g}^{\perp}$. This implies that the reduced Courant algebroid $E'_{\g}$ is 
\begin{equation} \label{eq_KdKdperp}
E'_{\g} = \frac{K^{\perp}_{\g} / G}{(K_{\g} \cap K^{\perp}_{\g}) / G} = \frac{K^{\perp}_{\g} / G}{K_{\g} / G}.
\end{equation}
Using the maximal isotropy of $\g$ with respect to $\<\cdot,\cdot\>_{\d}$, one can rewrite this quotient. As it was shown in \cite{2015LMaPh.tmp...53S}, this is an application of an elementary linear algebra:
\begin{lemma}
Let $(\d,\g)$ be a Manin pair. Let $E = V \oplus V^{\ast}$, where $V$ is a finite-dimensional vector space. Let $\<\cdot,\cdot\>_{E}$ be the canonical pairing on $E$. Let $\Re: \d \rightarrow E$ be an injective linear anti-isometry, that is $\< \Re(x), \Re(y) \>_{E} = - \<x,y\>_{\d}$ for all $x,y \in \d$. Let $K_{\d} = \Re(\d)$ and $K_{\g} = \Re(\g)$. Then there is a subspace decomposition:
\begin{equation} \label{eq_Kperpgdecomp}
K^{\perp}_{\g} = K_{\g} \oplus K^{\perp}_{\d}. 
\end{equation}
\end{lemma}
\begin{proof}
Clearly $K_{\d}^{\perp} \subseteq K_{\g}^{\perp}$ and $K_{\g} \subseteq K_{\g}^{\perp}$ due to the isotropy. This shows that $K_{\g} + K_{\d}^{\perp} \subseteq K_{\g}^{\perp}$. Moreover, one has $K_{\g} \cap K^{\perp}_{\d} = 0$. Any $\psi \in K_{\g} \cap K^{\perp}_{\d}$ can be written as $\psi = \Re(x)$ for $x \in \g$. On the other hand, we have $0 = \< \psi, \Re(y) \>_{E} = -\<x,y\>_{\d}$ for all $y \in \d$. The non-degeneracy of $\<\cdot,\cdot\>_{\d}$ implies $x = 0$. Hence the sum is direct and $K_{\g} \oplus K^{\perp}_{\d} \subseteq K^{\perp}_{\g}$. Finally, one has
\begin{equation}
\dim( K_{\g} \oplus K^{\perp}_{\d}) = \dim(K_{\g}) + (\dim(E) - \dim(K_{\d})) = \dim(E) - \dim(K_{\g}) = \dim(K_{\g}^{\perp}). 
\end{equation}
We have used the injectivity of $\Re$ and the maximal isotropy implication $\dim{\d} = 2 \cdot \dim{\g}$.
\end{proof}
This lemma is naturally generalized to vector bundles and our $\Re: \d \rightarrow \Gamma(\gTD)$ satisfies the assumptions of the lemma. In particular, we obtain the decomposition of the corresponding $\cif$-modules of $G$-invariant sections in the form 
\begin{equation}
\Gamma_{G}(K^{\perp}_{\g}) = \Gamma_{G}(K_{\g}) \oplus \Gamma_{G}(K^{\perp}_{\d}). 
\end{equation}
This makes the calculation of the above quotient trivial and we conclude that 
\begin{equation} \label{eq_EprimegisKperpd}
\Gamma(E'_{\g}) \cong \Gamma_{G}(K^{\perp}_{\d}). 
\end{equation}
From the expressions in (\ref{eq_KdKdperpexpl}) it is now straightforward to find the $\cif$-module $\Gamma_{G}(K^{\perp}_{\d})$: 
\begin{equation}
\Gamma_{G}(K^{\perp}_{\d}) = \{ (\Phi, \frac{1}{2} g_{\d}(\Phi)) \; | \; \Phi \in C^{\infty}_{G}(D,\d) \} \cong C^{\infty}_{G}(D,\d).
\end{equation}
Now, note that $\cif$-module $C^{\infty}_{G}(D,\d)$ is isomorphic to $C^{\infty}(M,\d) = \Gamma(M \times \d)$. This already suggests that the reduced Courant algebroid $E'_{\g}$ corresponds to the one in Subsection \ref{subsec_LieCourant}. 

\begin{tvrz} \label{tvrz_reducedEiso}
The above described $\cif$-module isomorphism defines an isomorphism of the reduced Courant algebroid $(E'_{\g}, \rho'_{\g}, \<\cdot,\cdot\>_{E'_{\g}}, [\cdot,\cdot]_{E'_{\g}})$ and $(E,\rho,\<\cdot,\cdot\>_{E},[\cdot,\cdot]_{E})$, the Courant algebroid analyzed in Subsection \ref{subsec_LieCourant}.
\end{tvrz}
\begin{proof}
First, we have to prove that the anchor $\rho'_{\g} \in \Hom(E'_{\g}, TM)$ defined by (\ref{eq_rho'gdef}) corresponds to $\rho \in \Hom(E,TM)$ defined by (\ref{eq_anchoronM}). As both maps are vector bundle morphisms, it suffices to verify the equality on constant sections of $E'_{\g}$. Let $\Phi_{x} \in C^{\infty}_{G}(D,\d)$ be the constant function corresponding to $x \in \d$. Under the above isomorphism, the constant sections $(\Phi_{x}, \frac{1}{2} g_{\d}(\Phi_{x})) \in \Gamma_{G}(K^{\perp}_{\d})$ correspond to constant sections $\psi_{x} \in \Gamma(E)$. For any $x \in \d$, we have
\begin{equation}
\{ \rho'_{\g}( \Phi_{x}, \frac{1}{2} g_{\d}(\Phi_{x})). f\} \circ \pi = \rho_{\d}( \Phi_{x}, \frac{1}{2} g_{\d}(\Phi_{x})).(f \circ \pi) = \frD_{\Phi_{x}}(f \circ \pi) = x^{R}.(f \circ \pi),
\end{equation}
for all $f \in \cif$. We thus have to show that $(\#^{\tr}(x).f) \circ \pi = x^{R}.(f \circ \pi)$, where $\#^{\tr}(x)$ are infinitesimal generators of the left action $\tr$ defined above (\ref{eq_anchoronM}). However, this easily follows from the definition. We have thus proved the relation $\rho'_{\g}( \Phi_{x}, \frac{1}{2} g_{\d}(\Phi_{x})) = \rho(\psi_{x})$. As constant sections generate the respective $\cif$-modules, this proves the assertion for the anchor maps. 

For the bracket and the pairing, observe that in (\ref{eq_Ebracketonconstant}) we have  proved that on constant sections the bracket gives $-[\cdot,\cdot]_{\d}$. This is exactly what (\ref{eq_Ebracket}) gives on constant sections of $E$. Similar argument and (\ref{eq_Apairingonconstant}) prove the correspondence of the two fiber-wise metrics. 
\end{proof}
\subsection{Moving the generalized metric: PLT duality} \label{subsec_PLT}
Recall that Poisson--Lie T-duality is a correspondence of two sets of sigma model backgrounds $(g,B)$ and $(\~g,\~B)$ on mutually dual Poisson Lie groups $G$ and $G^{\ast}$, respectively. See the original papers \cite{Klimcik:1995jn, Klimcik:1995ux, Klimcik:1995dy}. It was suggested already in \cite{Severa:2017oew} and more recently in \cite{2015LMaPh.tmp...53S} that this in fact corresponds to a lift and subsequent reductions of certain subbundles in the reduction scheme described in the previous subsection. We will now review this procedure. 

First, let us assume that we have a Manin triple $(\d,\g,\g')$. We will view $\d$ as the double $\d = \g \oplus \g^{\ast}$ where $\<\cdot,\cdot\>_{\d}$ is identified with the canonical pairing of $\g^{\ast}$ and $\g$ and the bracket has the form (\ref{eq_bracketdouble}). See Subsection \ref{subsec_LiePoisson} for details. The generalized metric on $(\d,\<\cdot,\cdot\>_{\d})$ corresponds to the choice of a $m$-dimensional positive definite vector subspace $\E \subseteq \d$, where $m = \dim{\g}$. Hence
\begin{equation}
\d = \E \oplus \E^{\perp}.
\end{equation} 
If necessary, we will write $\E_{+} \equiv \E$ and $\E_{-} \equiv \E^{\perp}$. Similarly to the generalized tangent bundle and Example \ref{ex_genmetricgTM}, there exists a vector space isomorphism $E_{0} \in \Hom(\g,\g^{\ast})$, such that 
\begin{equation} \label{eq_EasE0}
\E = \{ (x, E_{0}(x)) \; | \; x \in \g \} \subseteq \d,
\end{equation}
and $E_{0} = g_{0} + B_{0}$ for a positive definite metric $g_{0}$ on $\g$ and $B_{0} \in \Lambda^{2} \g^{\ast}$. Now, recall that we have obtained $\d$ as a reduced vector bundle $\d \cong E'_{\d} = K^{\perp}_{\d} / D$. This implies that the vector bundle $K^{\perp}_{\d}$ can be viewed as a pullback vector bundle $K^{\perp}_{\d} = \pi^{!}_{D}(K^{\perp}_{\d} / D)$. In particular, the subbundle $\E \subseteq \d$ can be pulled back to a $D$-invariant subbundle $\E_{D} \subseteq K^{\perp}_{\d}$. Explicitly, viewing $K^{\perp}_{\d}$ as a subbundle of the trivial vector bundle\footnote{See the previous subsection.} $D \times (\d \oplus \d^{\ast})$, one finds 
\begin{equation} \label{eq_EDexpl}
\Gamma(\E_{D}) = \{(\Phi, \frac{1}{2} g_{\d}(\Phi)) \; | \; \Phi \in C^{\infty}(D,\E) \}.
\end{equation}
But $K^{\perp}_{\d} \subseteq K^{\perp}_{\g}$ and we can thus view $\E_{D}$ as a $G$-invariant subbundle $\E_{D} \subseteq K^{\perp}_{\g}$. One can thus form a corresponding reduced subbundle $\E'_{\g} \equiv \E_{D} / G \subseteq E'_{\g}$. In other words, we have
\begin{equation} \label{eq_Eprimeggmetric}
\Gamma(\E'_{\g}) \cong \Gamma_{G}(\E_{D}) = \{ ( \Phi, \frac{1}{2} g_{\d}(\Phi)) \; | \; \Phi \in C^{\infty}_{G}(D,\E) \}. 
\end{equation}
When we view $\E'_{\g}$ as a subbundle of the trivial vector bundle $E = M \times \d$, it follows that $\E'_{\g}$ is in fact the trivial subbundle $\E'_{\g} = M \times \E$ of $E$, and the corresponding orthogonal decomposition is:
\begin{equation} \label{eq_Edecompgenmetric}
E = \E'_{\g} \oplus \E'^{\perp}_{\g} = (M \times \E) \oplus (M \times \E^{\perp}). 
\end{equation}
As the pairing on $E$ is the constant fiber-wise Lie algebra metric $\<\cdot,\cdot\>_{\d}$, it follows that $\E'_{\g} \subseteq E$ is a maximal positive definite subbundle and thus defines a generalized metric on $E$.

As in Subsection \ref{subsec_LiePoisson}, we can now assume that $\g^{\ast}$ is a Lie algebra of a connected closed Lie subgroup $G^{\ast} \subseteq D$ and the infinitesimal dressing action integrates to the global one. It thus induces a diffeomorphism $M \cong G^{\ast}$, where $(G^{\ast}, \Pi^{\ast})$ is the Poisson--Lie group dual to $(G,\Pi)$. Under this correspondence, $\E'_{\g}$  can be viewed as a trivial subbundle $G^{\ast} \times \E$ of $E' = G^{\ast} \times \d$. Moreover, the anchor $\rho' \in \Hom(E',TG^{\ast})$ given by (\ref{eq_rho'anchor}) and the isotropic splitting $s' \in \Hom(TG^{\ast}, E')$ defined by (\ref{eq_s'splitting}) provide for an isomorphism $\fPsi \in \Hom(\gTGa,E')$ having the form 
\begin{equation} \label{eq_fPsiisom}
\fPsi(X,\beta) := s'(X) + \rho'^{\ast}(\beta),
\end{equation}
for all $(X,\beta) \in \Gamma(\gTGa)$. The next step is suggestive. We can use $\fPsi$ to induce a generalized metric $V^{\g}_{+} \subseteq \gTGa$. It follows from Example \ref{ex_genmetricgTM} that there is a unique pair $(\~g,\~B)$ consisting of a Riemannian metric $\~g$ on $G^{\ast}$ and a $2$-form  $\~B \in \Omega^{2}(G^{\ast})$. Hence, let  $V^{\g}_{+} := \fPsi^{-1}(\E'_{\g})$. 

To find $\~g$ and $\~B$, it is convenient to work in terms of the right trivialization of $TG^{\ast}$. In other words, introduce $\~\fg \in C^{\infty}(G^{\ast}, S^{2} \g)$ and $\~\fB \in C^{\infty}(G^{\ast}, \Lambda^{2} \g)$ and $\~\fE \in C^{\infty}(G^{\ast}, \Hom(\g^{\ast},\g))$ via
\begin{equation}
\~\fg(\xi,\eta) := \~g( \xi_{R'}, \eta_{R'}), \; \; \~\fB(\xi,\eta) := \~B(\xi_{R'}, \eta_{R'}), \; \; \~\fE := \~\fg + \~\fB,
\end{equation}
for all $\xi,\eta \in \g^{\ast}$. We will sometimes abuse the notation and view for example $\~\fE$ simply as $G^{\ast}$-dependent map between two vector spaces $\g^{\ast}$ and $\g$. Similarly, we can define $\fPi^{\ast} \in C^{\infty}(G^{\ast}, \Lambda^{2} \g^{\ast})$. Note that the map $\rho'^{\ast}$ appearing in $\fPsi$ can be written as 
\begin{equation} \label{eq_rhoastexpl}
\rho'^{\ast}( x^{R'}) = \psi_{(x,0)} - \fPi^{\ast}(t_{k},x) \cdot \psi_{(0, t^{k})},
\end{equation}
where $\{ t_{k} \}_{k=1}^{\dim{\d}}$ is a fixed but arbitrary basis of $\g$. Let $(\xi_{R'}, \~\dE(\xi_{R'})) \in \Gamma(V^{\g}_{+})$, where $\~\dE := \~g + \~B$. Its image under the isomorphism $\fPsi \in \Hom(\gTGa,E')$ is a section of $\E'_{\g} = G^{\ast} \times \E$. There is thus a unique $\g$-valued function $\gamma \in C^{\infty}(G^{\ast}, \g)$, such that
\begin{equation}
\fPsi( \xi_{R'}, \~\dE(\xi_{R'})) = \gamma^{k} \cdot \psi_{(t_{k},E_{0}(t_{k}))} 
\end{equation} 
Evaluating $\fPsi$ on the left-hand side using the explicit forms for $s'$ and $\rho'^{\ast}$, one finds 
\begin{equation}
\fPsi( \xi_{R'}, \~\dE(\xi_{R'})) = \< t^{k}, \~\fE(\xi) \> \cdot \psi_{(t_{k},0)} + \{ \xi^{k} - \fPi^{\ast}(t_{k},\~\fE(\xi)) \} \cdot \psi_{(0,t^{k})}. 
\end{equation}
This immediately yields $\gamma = \~\fE(\xi)$. Examining the $\g^{\ast}$-component, we find the expression
\begin{equation} \label{eq_EastasE0andPi}
\~\fE = (E_{0} + \fPi^{\ast})^{-1}. 
\end{equation}
This is precisely the relation giving the background $\~\dE = \~g + \~B$ of the sigma model on the Lie group $G^{\ast}$ as proposed in \cite{1995PhLB..351..455K}. Completely analogously, one can obtain the dual fields $\dE = g + B$ on the Lie group $G$ corresponding to the generalized metric $V^{\g^{\ast}}_{+} \subseteq \gTG$. The construction of $V^{\g^{\ast}}_{+}$ is completely analogous to the one of $V^{\g}_{+}$ and uses \textit{the same given subspace} $\E \subseteq \d$. Explicitly, one finds $\fE = (E_{0}^{-1} + \fPi)^{-1}$, where $\fE = \fg + \fB$. 

\subsection{Generalized metric upstairs} \label{subsec_genmetupstairs}
Note that the lifted subbundle $\E_{D} \subseteq E_{\d} \cong \gTD$ given by (\ref{eq_EDexpl}) is positive definite, but it is not maximal as $\rk(\E_{D}) = \frac{1}{2} \dim{\g} < \dim{\g}$. $\E_{D}$ is thus not a generalized metric on $(E_{\d}, \<\cdot,\cdot\>_{E_{\d}})$. However, we may still try to find a generalized  metric $V^{\d}_{+} \subseteq E_{\d}$, such that $\E_{D} \subseteq V_{+}^{\d}$ and the generalized metric $\E \subseteq \d \cong E'_{\d}$ we have started with is obtained naturally via the reduction from the upstairs vector bundle $E_{\d}$. 

Recall that $E_{\d} = K_{\d} \oplus K_{\d}^{\perp}$ and $\E_{D} \subseteq K^{\perp}_{\d}$. We thus seek for a positive definite subbundle $\E'_{D} \subseteq K_{\d}$ such that $\rk(\E'_{D}) = \frac{1}{2} \dim{\g}$, in order to define $V^{\d}_{+} := \E_{D} \oplus \E'_{D}$. Also, recall that 
\begin{equation}
\Gamma(K_{\d}) = \{ (\Phi, -\frac{1}{2} g_{\d}(\Phi)) \; | \; \Phi \in C^{\infty}(D,\d) \}. 
\end{equation}
Looking at (\ref{eq_EDexpl}) and noting the opposite sign in the second component suggests the possible definition of the subbundle $\E'_{D}$. Namely
\begin{equation}
\Gamma(\E'_{D}) := \{ ( \Phi, - \frac{1}{2} g_{\d}(\Phi)) \; | \; \Phi \in C^{\infty}(D, \E^{\perp}) \}.
\end{equation}
Clearly $\rk(\E'_{D}) = \frac{1}{2} \dim{\g}$, as required. Moreover, as $\E^{\perp}$ is a negative definite subspace with respect to $\<\cdot,\cdot\>_{\d}$, it follows that $\E'_{D}$ is positive definite. Whence $V^{\d}_{+} \subseteq E_{\d}$ is a generalized metric. 

One interesting question arises. According to Example \ref{ex_genmetricgTM}, the generalized metric $V_{+}^{\d} \subseteq E_{\d} \cong \gTD$ uniquely determines a pair of fields $(g^{\d}, B^{\d})$ on $D$. What are they in this case? Both $\E_{D}$ and $\E'_{D}$ are $D$-invariant subbundles and so is $V^{\d}_{+}$. But this implies that there is a map $\fA \in \Hom(\d,\d^{\ast})$ of vector spaces, such that  $\Gamma(V_{+}^{\d}) = \{ (\Phi, \fA(\Phi)) \; | \; \Phi \in C^{\infty}(D,\d) \}$. Since $\E$ and $\E^{\perp}$ are fiber-wise constant, we can determine the result on the level of vector spaces. We have to find an $\fA$ so that, for any given $x \in \d$, there is a unique solution $y \in \E$ and $z \in \E^{\perp}$ to the equation
\begin{equation}
(x,\fA(x)) = (y + z, \frac{1}{2} g_{\d}(y - z)),
\end{equation}
As $\d = \g \oplus \g^{\ast}$, we write $x = (a,\xi)$ and $y = (b, E_{0}(b))$ and $z = (c, -E_{0}^{T}(c))$, for $a,b,c \in \g$ and $\xi \in \g^{\ast}$, where we have used the map $E_{0}$ defined by (\ref{eq_EasE0}).  
The condition $x = y + z$ leads to a simple system of linear equations for $(b,c)$ with the solution:
\begin{equation}
b = \frac{1}{2}\{ a + g_{0}^{-1}(\xi - B_{0}(a))\}, \; \; c = \frac{1}{2} \{a - g_{0}^{-1}(\xi -B_{0}(a)) \} .
\end{equation}
Plugging these expressions via $(y,z)$ into the equation $\fA(x) = \frac{1}{2} g_{\d}(y-z)$ gives the unique result for the map $\fA$. Writing it as a formal $2 \times 2$ matrix from $\d = \g \oplus \g^{\ast}$ to $\d^{\ast} = \g^{\ast} \oplus \g$ leads to 
\begin{equation}
\fA = \frac{1}{2} \bma{g_{0} - B_{0} g_{0}^{-1} B_{0}}{B_{0} g_{0}^{-1}}{-g_{0}^{-1} B_{0}}{g_{0}^{-1}} = \frac{1}{2} \gm_{0},
\end{equation}
where $\gm_{0} \in S^{2}(\d^{\ast})$  is the (fiber-wise) positive-definite metric corresponding to the original generalized metric $\E \subseteq \d$. 

We see that the map $\fA \in \Hom(\d,\d^{\ast})$ has the trivial skew-symmetric part, which immediately yields $B^{\d} = 0$. On the other hand, the metric $g^{\d}$ is the right-invariant metric on $D$ generated by the symmetric form $\frac{1}{2} \gm_{0}$ at the group unit.  

Finally, let $\tau^{\d} \in \Aut(E_{\d})$ be the involution corresponding to the generalized metric $V^{\d}_{+} \subseteq E_{\d}$. It is a straightforward calculation to see that $\tau^{\d}$ is block-diagonal with respect to the splitting $E_{\d} = K_{\d} \oplus K^{\perp}_{\d}$. Moreover, as $V_{+}^{\d}$ is $D$-invariant, it immediately follows that $[\Re(x), \tau^{\d}(\Phi,\Psi)]_{E_{\d}} = \tau^{\d}([\Re(x),(\Phi,\Psi)]_{E_{\d}})$ for all $x \in \d$ and $(\Phi,\Psi) \in \Gamma(E_{\d})$. These two properties imply that one can naturally restrict $\tau^{\d}$ to the $C^{\infty}(\{ \ast \})$-module $\Gamma(E'_{\d} ) = \Gamma_{D}(K^{\perp}_{\d}) \cong \d$ to obtain the orthogonal involution $\tau \in \Aut(\d)$ corresponding to the original generalized metric $\E \subseteq \d$. Similarly, by restricting $\tau^{\d}$ to the $C^{\infty}(M)$-module $\Gamma(E'_{\g}) \cong \Gamma_{G}(K^{\perp}_{\d})$, one obtains the orthogonal involution $\tau'_{\g} \in \Aut(E'_{\g})$ corresponding to the generalized metric $\E'_{\g}$. 
\section{Connections and the dilaton} \label{sec_dilatons}
We have seen that the lift and the consequent reduction of the generalized metric using the Lagrangian subgroups $G$ and $G^{\ast}$ gives us the correct formula for mutually Poisson--Lie T-dual backgrounds $\dE$ and $\~\dE$. The main idea of this section is to obtain a pair of Levi-Civita connections $\cD \in \LC(\gTG, V_{+}^{\g^{\ast}})$ and $\~\cD \in \LC(\gTGa, V_{+}^{\g})$, both of them satisfying the assumptions of Theorem \ref{thm_eomconn}. This will make them suitable to describe the equations of motion of the respective string effective actions. 

\subsection{Lift of the connection}
First, note that it follows from (\ref{eq_KdKdperpexpl}) that there is an isomorphism $\fPsi_{\w}$ of the vector space $\Gamma_{D}(E_{\d})$ with the double $\w = \d \oplus \d$, given by 
\begin{equation} \label{eq_fPsiw} \fPsi_{\w}(x,y) = (\Phi_{x+y}, \Psi_{ \frac{1}{2} g_{\d}(x-y)}), \end{equation}
for all $(x,y) \in \w$. In particular, $\fPsi_{\w}(\d \oplus \{ 0 \} ) = \Gamma_{D}(K^{\perp}_{\d})$ and $\fPsi_{\w}( \{0\} \oplus \d) = \Gamma_{D}(K_{\d})$. Recall that $\Gamma_{D}(E_{\d}) = \Gamma_{D}(K^{\perp}_{\d}) \oplus \Gamma_{D}(K_{\d})$ forms an involutive subspace of $\Gamma(E_{\d})$ with respect to the bracket (\ref{eq_HDorfmantrivial}). We can thus induce a Courant algebroid structure $(\w, 0, \<\cdot,\cdot\>_{\w}, [\cdot,\cdot]_{\w})$ using the isomorphism (\ref{eq_fPsiw}). Explicitly, one finds 
\begin{align}
\< (x,y), (x',y') \>_{\w} = & \ \<x,x'\>_{\d} - \<y,y'\>_{\d}, \\
[ (x,y), (x',y') ]_{\w} = & \ (-[x,x']_{\d} + [y,y']_{\d}, -[x,y']_{\d} + [x',y]_{\d} - 2 [y,y']_{\d}),
\end{align}
for all $(x,y), (x',y') \in \w$. Observe that the subspace $\d \oplus \{0\}$ corresponding to $\Gamma_{D}(K^{\perp}_{\d})$ is indeed involutive with respect to $[\cdot,\cdot]_{\w}$ whereas $\{0\} \oplus \d$ corresponding to $\Gamma_{D}(K_{\d})$ is not. The subspace $\d \oplus \{0\}$ by definition corresponds to 
the reduced Courant algebroid $(\d,0,\<\cdot,\cdot\>_{\d}, -[\cdot,\cdot]_{\d})$ as constructed in Subsection \ref{subsec_Dreduction}. 

Let $\cD^{0} \in \LC(\d, \E)$ be any Levi-Civita connection on $\d$. We aim to construct a connection $\cD^{\d} \in \LC(E_{\d},V_{+}^{\d})$ which will, in the most natural way, reduce to $\cD^{0}$. Recall that $E_{\d} = D \times (\d \oplus \d^{\ast}) \cong \gTD$ and $V_{+}^{\d} \subseteq E_{\d}$ is the lifted generalized metric constructed in Subsection \ref{subsec_genmetupstairs}. In particular, $\cD^{\d}$ has to be a $D$-invariant connection. This can be rewritten in terms of the Courant algebroid action 
map $\Re: \d \rightarrow \Gamma(E_{\d})$ for all $x \in \d$ as 
\begin{equation}
[ \Re(x), \cD^{\d}_{(\Phi,\Psi)} (\Phi',\Psi') ]_{E_{\d}} = \cD^{\d}_{ [\Re(x), (\Phi,\Psi)]_{E_{\d}}} (\Phi',\Psi') + \cD^{\d}_{(\Phi,\Psi)} [\Re(x), (\Phi',\Psi')]_{E_{\d}},
\end{equation}
for all $(\Phi,\Psi),(\Phi',\Psi') \in \Gamma(E_{\d})$. 

Equivalently, we may consider the connection  $\cD^{\w} \in \LC(\w, \E_{\w})$, where $\E_{\w} = \fPsi^{-1}_{\w}( \Gamma_{D}(V^{\d}_{+}))$ is the generalized metric induced on $\w$ from the lifted generalized metric $V^{\d}_{+}$. Moreover, $\cD^{\w}$ restricted onto $\d \oplus \{0\}$ has to give the starting connection $\cD^{0}$. With respect to the splitting $\w = \d \oplus \d$, the most general  Courant algebroid connection restricting to $\cD^{0}$ has the block form 
\begin{equation}
\cD^{\w}_{(x,0)} = \bma{\cDN_{x}}{0}{0}{\cD^{1}_{x}}, \; \; \cD^{\w}_{(0,y)} = \bma{ \cD^{2}_{y}}{A'_{y}}{g_{\d}^{-1} A'^{T}_{y} g_{\d}}{ \cD^{3}_{y}},
\end{equation}
where $\cD^{0},\cD^{1},\cD^{2}$ and $\cD^{4}$ are all connections on $\d$ compatible with the pairing $\<\cdot,\cdot\>_{\d}$, and $\<A'_{x}(y),z\>_{\d} \equiv A'(x,y,z)$ for an arbitrary covariant tensor $A' \in \T_{3}^{0}(\d)$. By imposing the condition $T_{\cD^{\w}} = 0$ and by using $T_{\cD^{0}} = 0$, one immediately obtains $\cD^{2} = 0$ and 
\begin{align}
0 & = \< \cD^{3}_{x}y - \cD^{3}_{y}x + 2[x,y]_{\d}, z \>_{\d} + \< \cD^{3}_{z}x, y \>_{\d}, \\
0 & = \< A'_{y}(y') - A'_{y'}(y) - [y,y']_{\d}, x \>_{\d} - \< \cD_{x}^{1}y, y'\>_{\d},
\end{align}
for all $x,y,z \in \d$. The first condition shows that $\cD^{3}$ has to be torsion-free with respect to the bracket $-2 [\cdot,\cdot]_{\d}$, whereas the second equation shows that $\cD^{1}$ is for torsion-free $\cD^{\w}$ fully determined by the tensor $A'$ and the bracket $[\cdot,\cdot]_{\d}$. Note that such $\cD^{1}$ is automatically a Courant algebroid connection. 

Finally, we have to analyze the compatibility of $\cD^{\w}$ with the generalized metric $\E_{\w} \subseteq \w$. By recalling how we have constructed $V_{+}^{\d}$, it is easy to see that this subspace takes the form $\E_{\w} = (\E \oplus \{0\}) \oplus ( \{0\} \oplus \E^{\perp}) \subseteq \w$. Consequently, we obtain
\begin{equation}
\cD_{x}^{1}(\E^{\perp}) \subseteq \E^{\perp}, \; \; \cD^{3}_{y}(\E^{\perp}) \subseteq \E^{\perp}, \; \; A'_{y}(\E^{\perp}) \subseteq \E,
\end{equation}
for all $x,y \in \d$. First, the conditions imposed on $\cD^{3}$ imply that it can be written in the form 
\begin{equation}
\cD_{x}^{3}y = 2 \cdot \cD^{0}_{x}y + g_{\d}^{-1} \fra(x,y,\cdot),
\end{equation}
for all $x,y \in \d$, where $\fra \in \d^{\ast} \otimes \Lambda^{2} \d^{\ast}$ satisfies the constraint $\fra(x,y,z) + cyclic(x,y,z) = 0$ for all $x,y,z \in \d$. Recall that we have the decomposition $\d = \E_{+} \oplus \E_{-}$, where $\E_{+} \equiv \E$ and $\E_{-} = \E^{\perp}$. Let $x_{\pm} \in \E_{\pm}$ denote the two projections of the element $x \in \d$. The final condition on $\fra$ is that its only non-trivial components are $\fra(x_{+},y_{+},z_{+})$ and $\fra(x_{-},y_{-},z_{-})$. 

The condition $A'_{y}(\E^{\perp}) \subseteq \E$ implies that the components of the tensor $A'$ vanish for the following combinations: 
\begin{equation}
A'( x, y_{+},z_{+}) = A'( x, y_{-}, z_{-}) = 0,
\end{equation}
for all $x,y,z \in \d$. Finally, the condition $\cD^{1}_{x}( \E^{\perp}) \subseteq \E^{\perp}$ imposes the restriction
\begin{equation}
0 = \< \cD^{1}_{x}y_{+},z_{-} \>_{\d} = A'(y_{+},z_{-},x) - A'(z_{-},y_{+},x) - \< [y_{+},z_{-}]_{\d},x \>_{\d}. 
\end{equation}
This relation uniquely determines the $(+-+)$ and $(-+-)$ components of the tensor $A'$, namely
\begin{equation}
A'(x_{+},y_{-},z_{+}) = \< [x_{+},y_{-}]_{\d}, z_{+} \>_{\d}, \; \; A'(x_{-},y_{+},z_{-}) = \< [x_{-},y_{+}]_{\d}, z_{-} \>_{\d}, 
\end{equation}
for all $x,y,z \in \d$. The only freedom for $A'$ is thus in the $(++-)$ and $(--+)$ components. The corresponding Courant algebroid connection $\cD^{1}$ can be now, in terms of of $A'$, rewritten as
\begin{equation}
\begin{split}
\< \cD^{1}_{x}y, y' \>_{\d} = - \<[y_{+},y'_{+}]_{\d} + [y_{-},y'_{-}]_{\d}, x \>_{\d} & + A'(y_{-},y'_{-},x_{+}) - A'(y'_{-},y_{-},x_{+}) \\
& + A'(y_{+},y'_{+}, x_{-}) - A'(y'_{+}, y_{+}, x_{-}). 
\end{split}
\end{equation}
There is no further restriction on these components. We conclude that there exists a non-empty class of Levi-Civita connections $\cD^{\w} \in \LC(\w,\E_{\w})$ on the Courant algebroid $(\w, 0, \<\cdot,\cdot\>_{\w}, [\cdot,\cdot]_{\w})$ with respect to the generalized metric $\E_{\w} \subseteq \w$, such that each connection in this class reduces in a natural way to the given Levi-Civita connection $\cD^{0}$ on $\d$. Note that for the purposes of this paper, we do not need to fix the ambiguity in the connection $\cD^{\w}$. 

We can now use the vector space isomorphism $\fPsi_{\w}: \w \rightarrow \Gamma_{D}(E_{\d})$ to \textit{define} the connection $\cD^{\d} \in \LC(E_{\d}, V_{+}^{\d})$ first on constant sections in $\Gamma_{D}(E_{\d})$ and then extending it to the whole $C^{\infty}(D)$-module $\Gamma(E_{\d})$ in order to satisfy the rules in (\ref{eq_CAconnaxioms}). This is possible and to each $\cD^{\w} \in \LC(\w,\E_{\w})$ we find the the unique corresponding $D$-invariant connection $\cD^{\d} \in \LC(E_{\d},V_{+}^{\d})$ that reduces to $\cD^{0}$ when restricted onto $\Gamma_{D}(K^{\perp}_{\d})$. 
\subsection{Reduction of $\cD^{\d}$ to the Courant algebroid $E'_{\g}$} \label{subsec_conred}
Let $\cD^{\d}$ be an arbitrary lift of the given Levi-Civita connection $\cDN \in \LC(\d,\E)$. We will now argue that $\cD^{\d}$ can be naturally reduced to the Levi-Civita connection $\cD^{\g} \in \LC(E'_{\g}, \E'_{\g})$, where $\E'_{\g} \subseteq E'_{\g}$ is the generalized metric (\ref{eq_Eprimeggmetric}). Once more, we employ the important $\cif$-module isomorphism (\ref{eq_EprimegisKperpd}). In particular, we claim that $\cD^{\d}$ preserves the $\cif$-module of $G$-invariant sections $\Gamma_{G}(E_{\d})$. 

As it is obtained by Leibniz rule extension and $\Gamma_{G}(E_{\d}) \cong C^{\infty}_{G}(D,\d \oplus \d^{\ast})$, this boils down to proving that the vector field $\rho_{\d}( \Phi_{x}, \Psi_{\xi}) \in \X(D)$ preserves the $\cif$-module $C^{\infty}_{G}(D)$ for all $(x,\xi) \in \d \oplus \d^{\ast}$. Recall that  $\rho_{\d} \in \Hom(E_{\d},TD)$ is the anchor map defined in Subsection \ref{subsec_Dreduction}. Let $f \in C^{\infty}_{G}(D)$. For every $g \in G$ and $d \in D$, we obtain 
\begin{equation}
\begin{split}
\{ \rho_{\d}( \Phi_{x}, \Psi_{\xi}).f \}(d \cdot g) = & \ \{ x^{R}.f \}(d \cdot g) = \ddt f( \exp(tx) \cdot d \cdot g) = \ddt f( \exp(tx) \cdot d) \\
= & \ \{ x^{R}.f  \}(d) = \{ \rho_{\d}(\Phi_{x}, \Psi_{\xi}).f \}(d). 
\end{split}
\end{equation}
This proves that the resulting function $\rho_{\d}( \Phi_{x},\Psi_{\xi}).f$ is also $G$-invariant. We can thus restrict $\cD^{\d}$ to the $\cif$-module $\Gamma_{G}(E_{\d})$. Moreover, $\cD^{\d}$ was constructed to restrict onto the subbundle $K^{\perp}_{\d}$. As we have $\Gamma(E'_{\g}) \cong \Gamma_{G}(K^{\perp}_{\d})$, we obtain a natural restriction $\cD^{\g}: \Gamma(E'_{\g}) \times \Gamma(E'_{\g}) \rightarrow \Gamma(E'_{\g})$. 

Recall that we have the canonical identification $E'_{\g} \cong E \equiv M \times \d$ and can view $\E'_{\g}$ as the generalized metric $\E'_{\g} = M \times \E \subseteq E$, cf. Proposition \ref{tvrz_reducedEiso} and the equation (\ref{eq_Edecompgenmetric}). It is thus convenient to work with the connection $\cD^{\g}$ induced by this identification\footnote{As for the generalized metric, we denote it by the same symbol $\cD^{\g}$.} on $E$. Considering the isomorphism between $\Gamma(E)$ and $\Gamma_{G}(K^{\perp}_{\d})$ together with the isomorphism $\fPsi_{\w}$ of $\Gamma_{D}(E_{\d})$ and $\w$, it is straightforward to see that the constant section $\psi_{x} \in \Gamma(E)$ corresponds to the element $(x,0) \in \w$. It follows from the construction of $\cD^{\g}$ and $\cD^{\d}$ that there holds 
\begin{equation} \label{eq_congfinalformula}
\cD^{\g}_{\psi} \psi' = \cD^{0}_{\psi}\psi' + \rho(\psi).\psi', 
\end{equation}
where $\cD^{0}$ is a fiber-wise extension of the connection $\cD^{0}$ on $\d$ to the module $\Gamma(E) = C^{\infty}(M,\d)$, and $\rho(\psi).\psi'$ denotes the action of the vector field $\rho(\psi) \in $ on the $\d$-valued function $\psi'$. It is now very easy to prove the required properties of the connection $\cD^{\g}$. 

\begin{tvrz}
The $\R$-bilinear map $\cD^{\g}: \Gamma(E) \times \Gamma(E) \rightarrow \Gamma(E)$ given by (\ref{eq_congfinalformula}) defines a Courant algebroid connection on $(E,\rho,\<\cdot,\cdot\>_{\d},[\cdot,\cdot]_{E})$\footnote{See Subsection \ref{subsec_LieCourant} for the definition.}. Moreover, it is metric compatible with the generalized metric $\E'_{\g} \subseteq E$ and it is torsion-free, hence $\cD^{\g} \in \LC(E,\E'_{\g})$.
\end{tvrz}
\begin{proof}
Clearly, the map $\cD^{\g}$ defined by (\ref{eq_congfinalformula}) satisfies the rules for function multiplication. Its compatibility with both $g_{E}$ and $\E'_{\g}$ and vanishing of the torsion $3$-form  $T_{\cD^{\g}}$ are tensorial conditions and can be thus checked on constant sections. But this obviously boils down to the assumption $\cD^{0} \in \LC(\d,\E)$. 
\end{proof}
To conclude this subsection, note that we can consider the case where $\d = \g \oplus \g^{\ast}$ is a double of a Lie bialgebra $(\g,\delta)$. We can thus use the tools in Subsection \ref{subsec_LiePoisson} and in particular the isomorphism of $E$ with the Courant algebroid on $E' = G^{\ast} \times \d$. We can thus view $\cD^{\g}$ as a Levi-Civita connection on $E'$ with respect to $\E'_{\g} = G^{\ast} \times \E$. It is written using the similar formula
\begin{equation} \label{eq_cDgfinal}
\cD^{\g}_{\psi} \psi' = \cD^{0}_{\psi}\psi' + \rho'(\psi).\psi',
\end{equation}
for all $\psi,\psi' \in \Gamma(E')$, where $\rho' \in \Hom(E',TG^{\ast})$ is the anchor map given by (\ref{eq_rho'anchor}). 

\subsection{Splitting and the characteristic vector field $X_{\~\cD}$}
We are now getting closer to the main theorems of this paper. We will consider only the case described in Subsection \ref{subsec_LiePoisson}. In particular, we have the special splitting $s' \in \Hom(TG^{\ast}, E')$ which provides us with the explicit isomorphism $\fPsi \in \Hom( \gTG^{\ast}, E')$ given by (\ref{eq_fPsiisom}). In particular, we induce the generalized metric $V_{+}^{\g} := \fPsi^{-1}( \E'_{\g}) \subseteq \gTGa$ corresponding to the pair of fields $(\~g,\~B)$ as described in Subsection \ref{subsec_PLT}. This leads to the idea to define a Courant algebroid connection $\~\cD$ on $\gTGa$ as 
\begin{equation} \label{eq_cDastdef}
\~\cD_{(X,\xi)}(Y,\eta) := \fPsi^{-1}( \cD^{\g}_{\fPsi(X,\xi)} \fPsi(Y,\eta)), 
\end{equation}
for all $(X,\xi), (Y,\eta) \in \Gamma(\gTGa)$. As $\fPsi$ forms a Courant algebroid isomorphism of $E'$ and the (non-twisted) Dorfman bracket on $\gTGa$, it follows that $\~\cD \in \LC(\gTGa, V_{+}^{\g})$. 

Our intention is to use $\~\cD$ in order to describe the equations of motion of the low-energy string effective action. Therefore, we have to meet the assumptions of the Theorem \ref{thm_eomconn}. In particular, the characteristic vector field $X_{\~\cD}$ defined by (\ref{eq_charverfield}) must vanish. We will now prove that this imposes some non-trivial limitations on the choice of the connection $\cD^{0}$ and Lie bialgebra $(\g,\delta)$.  

First, note that we can calculate it directly in terms of the connection $\cD^{\g}$ as the characteristic vector field is invariant under Courant algebroid isomorphisms, see \cite{Jurco:2016emw} for clarification of this statement. We can thus calculate $X_{\~\cD}$ using the formula 
\begin{equation}
X_{\~\cD}.f = \Div_{\cD^{\g}}( \D'{f}) = \< \cD^{\g}_{\psi_{\alpha}} \D'{f}, \psi^{\alpha} \>, 
\end{equation}
where $\D' = \rho'^{\ast} \circ d$ and $\{ \psi_{\alpha} \}_{\alpha =1}^{\dim{\d}}$ are the constant generators of $\Gamma(E')$ corresponding to a fixed (but otherwise arbitrary) basis $\{ t_{\alpha} \}_{\alpha=1}^{\dim{\d}}$ of Lie algebra $\d$. Define an element $(\ft,\ftau) \in \d$ using the divergence operator for $\cD^{0}$ as 
\begin{equation} \label{eq_DivcD0}
\Div_{\cD^{0}}(x,\xi) = \< (\ft,\ftau), (x,\xi) \>_{\d},
\end{equation}
for all $(x,\xi) \in \d = \g \oplus \g^{\ast}$. One can now proceed with the calculation of $X_{\~\cD}$ to find
\begin{equation}
\begin{split}
\< \cD^{\g}_{\psi_{\alpha}} \D'{f}, \psi^{\alpha} \> = & \ \< \cD^{\g}_{\psi_{\alpha}} \D'{f}, \psi^{\alpha}_{\d} \>_{\d} = \rho'(\psi_{\alpha}). \< \D'{f}, \psi^{\alpha}_{\d} \>_{\d} - \< \D'{f}, \cD^{0}_{\psi_{\alpha}} \psi^{\alpha}_{\d} \>_{\d} \\
= & \ \rho'(\psi_{\alpha}).( \rho'(\psi^{\alpha}_{\d}).f) + \rho'( \psi_{(\ft,\ftau)}).f,
\end{split}
\end{equation}
where $\psi_{(\ft,\ftau)} \in \Gamma(E')$ denotes the constant section corresponding to $(\ft,\ftau) \in \d$. Let us decompose $X_{\~\cD}$ as a sum of two vector fields
\begin{equation}
X_{\~\cD} = Y + \rho'(\psi_{(\ft,\ftau)}),
\end{equation}
where $Y \in \X(G^{\ast})$ will be now calculated using a little trick. Indeed, define a Lie algebra representation $\hat{\rho}: \d \rightarrow \End( C^{\infty}(G^{\ast}))$ by setting $\hat{\rho}(x,\xi) = - \rho'(\psi_{(x,\xi)})$ for all $(x,\xi) \in \d$. Extending it to the homomorphism from the universal enveloping algebra $\mathfrak{U}(\d)$, we find that $Y = \hat{\rho}( \Omega)$, where $\Omega \in \mathfrak{U}(\d)$ is the quadratic Casimir element corresponding to the bilinear form $\<\cdot,\cdot\>_{\d}$. Note that it is then interesting on its own that $Y$ is a vector field. It follows that $Y$ must commute with $\hat{\rho}(x,\xi)$ for all $(x,\xi) \in \d$. In particular, it commutes with all right-invariant vector fields on $G^{\ast}$ as $\hat{\rho}(0,\xi) = -\xi_{R'}$. This implies that $Y$ is left-invariant. It thus suffices to calculate its value at the group unit $e \in G^{\ast}$. Using a split basis $\{ t_{\alpha} \}_{\alpha=1}^{\dim{\d}} = \{ t_{k} \}_{k=1}^{\dim{\g}} \cup \{ t^{k} \}_{k=1}^{\dim{\g^{\ast}}}$, we arrive to the following expression:
\begin{equation}
\begin{split}
Y_{e} = & \ ( t^{k}_{R'}.\Pi^{\ast}( t_{m}^{R'}, t_{k}^{R'}) )_{e} \cdot t^{m} = ( \Li{t^{k}_{R'}} \Pi^{\ast} )_{e}(t_{m},t_{k}) \cdot t^{m} = \< t^{k}, [t_{m},t_{k}]_{\g} \> \cdot t^{m} \\
= & \ \< \ad^{\ast}_{t_{k}} t^{k}, t_{m} \> \cdot t^{m} = \falpha,
\end{split}
\end{equation}
where $\falpha \in \g^{\ast}$ is a $1$-form defined as $\falpha(x) = \Tr( \ad_{x})$, for all $x \in \g$. We conclude that 
\begin{equation} \label{eq_charverfinal}
X_{\~\cD} = \falpha_{L'} + \ftau_{R'} + \Pi^{\ast}(\ft^{R'}). 
\end{equation}
This vector field has to vanish everywhere. In particular, at the unit $e \in G^{\ast}$. As $\Pi^{\ast}_{e} = 0$, we find the necessary condition for $\~\cD$ to satisfy the assumptions of Theorem \ref{thm_eomconn}, namely
\begin{equation}
\falpha + \ftau = 0.
\end{equation}
By plugging this back into the condition $X_{\~\cD} = 0$, we obtain the condition 
\begin{equation} \label{eq_alphaalphaisPiasttR}
\falpha_{L'} - \falpha_{R'} = - \Pi^{\ast}(\ft^{R'}). 
\end{equation}
Note that on the left-hand side, there is a multiplicative vector field on $G^{\ast}$. To analyze this further, let us recall the following notion. 
\begin{definice}
Let $\varphi: \g \rightarrow \X(G^{\ast})$ be an $\R$-linear map. We say that $\varphi$ defines a \textbf{left-twisted multiplicative vector field}  on $G^{\ast}$ if for any $x \in \g$, one has 
\begin{equation}
\varphi(x)_{gh} = L_{g \ast}( \varphi( \Ad_{g^{-1}}^{\ast}(x))_{h}) + R_{h \ast}( \varphi(x)_{g}),
\end{equation}
for all $g,h \in G^{\ast}$. Note that $\varphi(x)_{e} = 0$ for all $x \in \g$. Moreover, the vector field $\varphi(x) \in \X(G^{\ast})$ is multiplicative if and only if $x \in \g$ is $\Ad^{\ast}$-invariant. 
\end{definice}
It follows from the multiplicativity of $\Pi^{\ast}$ that the map $\varphi(x) = - \Pi^{\ast}(x^{R'})$ is a left-twisted multiplicative vector field on $G^{\ast}$. As the left-hand side of (\ref{eq_alphaalphaisPiasttR}) is multiplicative, so has to  be the right-hand side. \textbf{This forces $\ft$ to be the $\Ad^{\ast}$-invariant element of $\g$}. Finally, two multiplicative vector fields are equal iff their intrinsic derivatives coincide, that is 
\begin{equation}
(\Li{\zeta_{R'}} ( \falpha_{L'} - \falpha_{R'}))_{e} = - \{ \Li{\zeta_{R'}} (\Pi^{\ast}(\ft^{R'})) \}_{e}
\end{equation}
The left-hand side is easily calculated to be $[\zeta,\falpha]_{\g^{\ast}}$. On the other hand, we get 
\begin{equation}
- \{ \Li{\zeta_{R'}}( \Pi^{\ast}(\ft^{R'})) \}_{e} = - (\Li{\zeta_{R'}} \Pi^{\ast})_{e}(t_{k}, \ft) \cdot t^{k} = - \< \zeta, [t_{k},\ft]_{\g} \> \cdot t^{k} = - \ad^{\ast}_{\ft} \zeta.
\end{equation}
This gives us the final condition on $\ft$ and $\falpha$ that has to be valid for $(x,\zeta) \in \d$:
\begin{equation} \label{eq_finalcondonXnabla}
\< [\zeta, \falpha]_{\g^{\ast}}, x \> + \< \zeta, [x,\ft]_{\g} \> = 0. 
\end{equation}
Let us summarize the above results in the form of a theorem. 
\begin{theorem} \label{thm_charverfield}
The characteristic vector field $X_{\~\cD} \in \X(G^{\ast})$ corresponding to the connection $\~\cD$ (\ref{eq_cDastdef}) can be written as (\ref{eq_charverfinal}). 

In particular, the condition $X_{\~\cD} = 0$ forces $\ftau + \falpha = 0$, the vector $\ft \in \g$ must be $\Ad^{\ast}$-invariant (with respect to the $G^{\ast}$-action) and the condition (\ref{eq_finalcondonXnabla}) must hold.
\end{theorem}
\begin{example} \label{ex_unimod}
Let us assume that the Lie algebra $(\g,[\cdot,\cdot]_{\g})$ is \textbf{unimodular}, that is, $\falpha = 0$.  By definition, unimodular Lie algebras have traceless operators of adjoint representation. Equivalently, the left and right Haar measures on the integrating Lie group $G$ coincide. We then have to choose $\cD^{0}$ so that $(\ft,\ftau) = (\ft,0)$. Moreover, the equation (\ref{eq_finalcondonXnabla}) then forces $\ft$ to be the central element of $\g$ which it $\Ad^{\ast}$-invariant (coadjoint representation of $G^{\ast}$ on $\g$). As an example, the connection (\ref{eq_cD0formula}) in Example \ref{ex_quadrconn} satisfies $(\ft,\ftau) = (0,0)$, thus forming a good candidate. 
\end{example}
\subsection{Calculating the dilaton}
Now observe that there is a second assumption in Theorem \ref{thm_eomconn} assuming the existence of a scalar function $\~\phi \in C^{\infty}(G^{\ast})$ related to the connection $\~\cD$ as in (\ref{eq_dilatonrel}). This condition can be again expressed using the connection $\cD^{\g}$ and the corresponding anchor $\rho' \in \Hom(E,TG^{\ast})$. We have to find a scalar function $\~\phi \in C^{\infty}(G^{\ast})$ satisfying the equation
\begin{equation} \label{eq_dtphi}
\< d\~\phi, Z \> = \< \cD^{\g}_{ \rho'^{\ast}( \~g(e^{k}))} \rho'^{\ast}( e_{k}), \rho'^{\ast}( \~g(Z)) \>_{\d},
\end{equation}
for all vector fields $Z \in \X(G^{\ast})$, where $\{ e^{k} \}_{k=1}^{\dim{G^{\ast}}}$ is some local frame on $G^{\ast}$. Note that it is a priori \textit{not clear at all} that the right-hand side defines an exact $1$-form on $G^{\ast}$ and for a general connection $\cD^{0}$ and Lie bialgebra $(\g,\delta)$ it \textit{does not}. Let us denote the right-hand side of (\ref{eq_dtphi}) as $W'(Z)$ for $W' \in \Omega^{1}(G^{\ast})$. See Section 6 of \cite{Jurco:2016emw} for an explanation of this notation. 

Now, note that one can choose the right-invariant frame $e^{k} = t^{k}_{R'}$, where $\{ t_{k} \}_{k=1}^{\dim{\g}}$ is a fixed basis of $\g$ and $\{ t^{k} \}_{k=1}^{\dim{\g^{\ast}}}$ is the dual one of $\g^{\ast}$. Define a map $\fai \in C^{\infty}(G^{\ast},\Hom(\g,\d))$ as 
\begin{equation}
\fai(x) = (x, - \fPi^{\ast}(x)),
\end{equation}
for all $x \in \g$. We can then rewrite the above expression using $\cD^{0}$ and $\~\fg$ defined in (\ref{subsec_PLT}), getting 
\begin{equation} \label{eq_dphizetaR'}
W'(\zeta_{R'}) = \< \cDN_{\fai(\~\fg(t^{k}))} \fai(t_{k}), \fai( \~\fg(\zeta)) \>_{\d}.
\end{equation}
It is now convenient to modify the connection as follows. Recall that the generalized metric $\E \subseteq \d$ forms the graph (\ref{eq_EasE0}) of the map $E_{0} = g_{0} + B_{0}$. We can now consider a twisted connection 
\begin{equation}
\hcD^{0}_{x} y = e^{-B_{0}}( \cDN_{e^{B_{0}}(x)} e^{B_{0}}(y)),
\end{equation}
for all $x,y \in \d$, where $e^{B_{0}}(z,\zeta) := (z, \zeta + B_{0}(z))$ for all $(z,\zeta) \in \d$. It follows that $\hcD^{0}$ is a Levi-Civita connection on the Courant algebroid $(\d, 0, \<\cdot,\cdot\>_{\d}, - [\cdot,\cdot]'_{\d})$ with respect to the generalized metric $\E^{0} \subseteq \d$, where $[\cdot,\cdot]'_{\d}$ is the Lie algebra bracket defined as 
\begin{equation} [x,y]'_{\d} = e^{-B_{0}} [ e^{B_{0}}(x), e^{B_{0}}(y)]_{\d}, \end{equation}
for all $x,y \in \d$ and $\E^{0} = \{ (x,g_{0}(x)) \; | \; x \in \g \}$. The simplicity of $\E^{0}$ is the main reason why to consider the twisted scenario. Let $\fj = e^{-B_{0}} \fai$. This map has the explicit form 
\begin{equation} \label{eq_fjmap}
\fj(x) = e^{-B_{0}}(x, - \fPi^{\ast}(x)) = (x, - ( \fPi^{\ast} + B_{0})(x)) =: (x, -B'_{0}(x)),
\end{equation}
for all $x \in \g$. Rewriting (\ref{eq_dphizetaR'}) using the twisted objects, for each $\zeta \in \g^{\ast}$, we find 
\begin{equation} \label{eq_WzetaRprefinal}
W'(\zeta_{R'}) = \< \hcD^{0}_{\fj(\~\fg(t^{k}))} \fj(t_{k}), \fj( \~\fg(\zeta)) \>_{\d}. 
\end{equation}
We will now calculate this quantity for a \textit{special choice} of the connection $\cD^{0}$, namely the one examined in Example \ref{ex_quadrconn}. As already noted in Example \ref{ex_unimod} in the previous subsection, this connection is divergence-free. We will show that the general case is then obtained by a simple modification of this case. First, note that $\hcD^{0}$ is given by the same formula (\ref{eq_cD0formula}) as $\cD^{0}$, that is
\begin{equation}
\begin{split}
\< \hcD^{0}_{x}y, z \>_{\d} = & \ - \< [x_{+},y_{-}]'_{\d}, z_{-} \>_{\d} - \< [x_{-},y_{+}]'_{\d}, z_{+} \>_{\d} \\
 & \ - \frac{1}{3} \< [x_{+},y_{+}]'_{\d}, z_{+} \>_{\d} 	- \frac{1}{3} \< [x_{-},y_{-}]'_{\d}, z_{-} \>_{\d},
\end{split}
\end{equation}
for all $x,y,z \in \d$. The projections $x_{\pm}$ are with respect to the generalized metric $\E^{0} \subseteq \d$. We formulate the result in the form of a theorem. 
\begin{theorem} \label{thm_main1}
Let $W' \in \Omega^{1}(G^{\ast})$ be a $1$-form defined by (\ref{eq_WzetaRprefinal}). Then it can be rewritten as 
\begin{equation} \label{eq_W'zetaRmain1}
W'(\zeta_{R'}) = \< \zeta, \frac{1}{2} [\~\fB(t^{k}),t_{k}]_{\g} + ad^{\ast}_{\fPi^{\ast}(t_{k})} \~\fB(t^{k}) \> + \frac{1}{2} \<\zeta, \fa + \~\fB \cdot \fAd(\falpha) \>,
\end{equation}
where $\{ t_{k} \}_{k=1}^{\dim{\g}}$ is a fixed (but arbitrary) basis of $\g$, $\~\fB$ and $\fPi^{\ast}$ are defined using the right-invariant vector fields as in Subsection \ref{subsec_PLT}, $(\fa,\falpha) \in \d$ are traces of adjoint representations of $\g$ and $\g^{\ast}$ as defined in the previous subsection, and $\fAd \in C^{\infty}(G^{\ast}, \End(\g^{\ast}))$ is defined as 
\begin{equation}
(\fAd(\xi))_{g} = \Ad_{g}(\xi),
\end{equation}
for all $g \in G^{\ast}$ and $\xi \in \g^{\ast}$. 
\end{theorem}
Before going to the actual proof of this theorem, let us recall one consequence of the Jacobi identity for the Poisson structure $\fPi^{\ast}$. 
\begin{lemma} \label{lem_fPiidentity}
The map $\fPi^{\ast} \in C^{\infty}( G^{\ast}, \Lambda^{2} \g^{\ast})$ corresponding to the Poisson bivector $\Pi^{\ast} \in \X(G^{\ast})$ satisfies the following identity, for all $x,y \in \g$,: 
\begin{equation} \label{eq_fPiidentity}
0 = \fPi^{\ast}([x,y]_{\g}) + [ \fPi^{\ast}(x), \fPi^{\ast}(y) ]_{\g^{\ast}}  - ( \ad^{\ast}_{x} \fPi^{\ast}(y) - \ad^{\ast}_{y} \fPi^{\ast}(x) ) - \fPi^{\ast}( \ad^{\ast}_{\fPi^{\ast}(x)} y - \ad^{\ast}_{\fPi^{\ast}(y)}x ). 
\end{equation}
All Lie algebra operations are assumed to be defined point-wise. Denote the $2$-form on the right-hand side as $\omega_{\fPi^{\ast}}$, that is, there holds $\omega_{\fPi^{\ast}}(x,y) = 0$, for all $x,y \in \g$. 
\end{lemma}
\begin{proof}
Recall the left infinitesimal dressing action (\ref{eq_leftdress}) and the corresponding antihomomorphism property (\ref{eq_leftdressahom}). Rewriting this identity using $\fPi^{\ast}$ and the useful relation
\begin{equation} \label{eq_useful}
\zeta_{R'}. \fPi^{\ast}(x,y) = \< \zeta, [x,y]_{\g} \> - \fPi^{\ast}( \ad^{\ast}_{\zeta}x, y) - \fPi^{\ast}( x, \ad^{\ast}_{\zeta}y)
\end{equation}
(valid for all $x,y \in \g$ and $\zeta \in \g^{\ast}$) gives precisely the equation (\ref{eq_fPiidentity}). 
\end{proof}
Next, note that the bracket $[\cdot,\cdot]'_{\d}$ can be written as $[x,y]'_{\d} = [x,y]_{\d} - g_{\d}^{-1}\H(x,y,\cdot)$, for all $x,y \in \d$ and a unique $3$-form $\H \in \Lambda^{3} \d^{\ast}$. Note that $\E^{0}_{\pm} = \fPsi_{\pm}^{0}(\g)$, where $\fPsi_{\pm}^{0} \in \Hom(\g,\d)$ read 
\begin{equation}
\fPsi_{\pm}^{0}(x) = (x, \pm g_{0}(x)), 
\end{equation}
for all $x \in \g$. We will use the notation $x_{\pm} = \fPsi_{\pm}^{0}(x'_{\pm})$, for each $x \in \d$. One can use the isomorphisms of $\g$ with $\E^{0}_{\pm}$ to induce new objects on $\g$. For example, one can define $\R$-bilinear brackets\footnote{In general, they are not skew-symmetric.} $[\cdot,\cdot]_{+-}^{\pm}$ on $\g$ in the most natural way from $[\cdot,\cdot]_{\d}$, namely 
\begin{equation}
[ \fPsi_{+}^{0}(x), \fPsi_{-}^{0}(y) ]_{\d} =: \fPsi_{+}^{0}( [x,y]_{+-}^{+}) + \fPsi_{-}^{0}( [x,y]^{-}_{+-}),
\end{equation}
for all $x,y \in \g$. Similarly, let $\H_{\pm \pm \pm}(x,y,z) := \H( \fPsi_{\pm}^{0}(x), \fPsi_{\pm}^{0}(y), \fPsi_{\pm}^{0}(z))$ for all $x,y,z \in \g$. One can use (\ref{eq_bracketdouble}) to calculate these objects. We will not write down the results here. Just note that $\H_{+--} = H_{0} + H_{2}$ and $\H_{-++} = H_{0} - H_{2}$, where 
\begin{align}
H_{0}(x,y,z) & = -B_{0}([x,y]_{\g},z) - \< [B_{0}(x), B_{0}(y)]_{\g^{\ast}}, z \> + cyclic(x,y,z), \\
H_{2}(x,y,z) & = ( \ad^{(2)}_{g_{0}(x)} B_{0})(y,z) + ( \ad^{(2)}_{g_{0}(y)} B_{0})(y,z) - ( \ad^{(2)}_{g_{0}(z)} B_{0})(x,y),
\end{align}
and where $\ad^{(2)}$ denotes the adjoint action of $\g^{\ast}$ on $\Lambda^{2} \g^{\ast}$, that is, for all $\zeta \in \g^{\ast}$ and $y,z \in \g$, one has
\begin{equation}
(ad^{(2)}_{\zeta} B_{0})(y,z) = - B_{0}( \ad^{\ast}_{\zeta} y, z) - B_{0}( y, \ad^{\ast}_{\zeta} z ).
\end{equation}
Now we have prepared all necessary notions for the proof of the main theorem. 

\begin{proof}[Proof of Theorem \ref{thm_main1}]
We will first prove slightly different version of the equation (\ref{eq_W'zetaRmain1}). In particular, we claim that $W'(\zeta_{R'})$ can be, for all $\zeta \in \g^{\ast}$, expressed as 
\begin{equation} \label{eq_W'zetaRmain2}
\begin{split}
W'(\zeta_{R'}) = & \ \< \zeta, \frac{1}{2} [\~\fB(t^{k}), t_{k}]_{\g} + ad^{\ast}_{\fPi^{\ast}(t_{k})} \~\fB(t^{k}) \> + \frac{1}{2} \< \zeta, \fa \> \\
& + \frac{1}{2} \< \~\fB(\zeta), [\fPi^{\ast}(t_{k}),t^{k}]_{\g^{\ast}} - \falpha + \fPi^{\ast}(\fa) \>.
\end{split}
\end{equation}
We start from the expression (\ref{eq_WzetaRprefinal}). First, note that only the symmetric part of the connection $\hcD^{0}$ contributes (due to the sum over $k$). It reads
\begin{equation}
\frac{1}{2} \< \hcD^{0}_{x}y + \hcD^{0}_{y}x, z \>_{\d} = \frac{1}{2} \< [x_{+},y_{-}]'_{\d} + [y_{+},x_{-}]'_{\d}, z_{+} - z_{-} \>_{\d},
\end{equation}
for all $x,y,z \in \d$. Moreover, both terms contribute equally and it thus suffices to calculate 
\begin{equation}
\< [x_{+},y_{-}]'_{\d}, z_{+} - z_{-} \>_{\d} .
\end{equation}
Write $z = (c,\lambda)$ for $(c,\lambda) \in \g \oplus \g^{\ast}$. Using the notation introduced above, one finds 
\begin{equation} \label{eq_proofhumus1}
\begin{split}
\< [x_{+},y_{-}]'_{\d}, z_{+} - z_{-} \>_{\d} = & \ g_{0}( [x'_{+},y'_{-}]_{\g}, c) - g_{0}^{-1}( [ g_{0}(x'_{+}), g_{0}(y'_{-}) ]_{\g^{\ast}}, \lambda) \\
& + g_{0}( \ad^{\ast}_{g_{0}(x'_{+})} y'_{+} + \ad^{\ast}_{g_{0}(y'_{-})} x'_{+}, c ) \\
& - g_{0}( \ad^{\ast}_{x'_{+}} g_{0}(y'_{-}) + \ad^{\ast}_{y'_{-}} g_{0}(x'_{+}), \lambda) \\
& - H_{0}( x'_{+},y'_{-}, g_{0}^{-1}(\lambda)) - \frac{1}{2} \{ H_{2}(x'_{+},c,y'_{-}) - H_{2}( y'_{-}, c, x'_{+}) \} \\
& + \frac{1}{2} \{ H_{2}(x'_{+}, g_{0}^{-1}(\lambda), y'_{-}) + H_{2}( y'_{-}, g_{0}^{-1}(\lambda), x'_{+}) \}.
\end{split}
\end{equation}
To calculate $W'(\zeta_{R'})$, we have to plug $x = \fj(\~\fg(t^{k}))$, $y = \fj(t_{k})$ and $z = \fj(\~\fg(\zeta))$. Using the explicit form (\ref{eq_fjmap}) of the map $\fj$ and the relation (\ref{eq_EastasE0andPi}) implying $\~\fB = - g_{0}^{-1} B'_{0} \~\fg$, we obtain 
\begin{equation} \label{eq_x'+y'-zplug}
x'_{+} = \frac{1}{2}(\~\fg + \~\fB)( t^{k}), \; \; y'_{-} = \frac{1}{2} (\~\fg - \~\fB)( \fg^{-1}(t_{k})), \; \; (c,\lambda) = (\~\fg(\zeta), g_{0} \~\fB(\zeta)). 
\end{equation}
Now, the most difficult part of the proof comes with plugging (\ref{eq_x'+y'-zplug}) into the right-hand side of (\ref{eq_proofhumus1}). However, as it is a straightforward calculation consisting mostly of a careful grouping of the terms coming from the substitution, we omit it here. Note that one has to use the identity $g_{0} \~\fg + B'_{0}\~\fB = \f1_{\g^{\ast}}$ coming from (\ref{eq_EastasE0andPi}). The result for $W'$ is 
\begin{equation}
\begin{split}
W'(\zeta_{R'}) = & \  \<\zeta, \frac{1}{2} [\~\fB(t^{k}), t_{k}]_{\g} + ad^{\ast}_{\fPi^{\ast}(t_{k})} \~\fB(t^{k}) \>  + \frac{1}{2}\< \zeta, \fa \>\\
& + \frac{1}{2} \< \~\fB(\zeta), [\fPi^{\ast}(t_{k}),t^{k}]_{\g^{\ast}} - \falpha + \fPi^{\ast}(\fa) \>  + \frac{1}{2} \< \~\fB(\zeta), \omega_{\fPi^{\ast}}( \~\fB(t^{k}), t_{k}) \>.
\end{split}
\end{equation}
At this very moment we employ Lemma \ref{lem_fPiidentity} and the equation (\ref{eq_fPiidentity}) to see that this is exactly the equation (\ref{eq_W'zetaRmain2}). To finish the proof, it remains to prove that it is equivalent to (\ref{eq_W'zetaRmain1}). To do so, we examine $\fb \in C^{\infty}(G^{\ast}, \g^{\ast})$ defined as $\fb := [ \fPi^{\ast}(t_{k}), t^{k} ]_{\g^{\ast}} + \fPi^{\ast}(\fa) $. First, we get
\begin{equation}
\begin{split}
[t^{k}_{R'}, \Pi^{\ast}(t_{k}^{R'}) ] = & \ (\Li{t^{k}_{R'}}\Pi^{\ast})(t_{k}^{R'}) + \Pi^{\ast}( \Li{t^{k}_{R'}} t_{k}^{R'}) = \< \falpha - \fPi^{\ast}(\fa), t_{m} \> \cdot t^{m}_{R'}
\end{split}
\end{equation}
On the other hand, one can write the same expression as 
\begin{equation}
[t^{k}_{R'}, \Pi^{\ast}(t_{k}^{R'})] = ( t^{k}_{R'}.\Pi^{\ast}(t_{m},t_{k}) ) \cdot t^{m}_{R'} + \< [ \fPi^{\ast}(t_{k}), t^{k}]_{\g^{\ast}}, t_{m} \> \cdot t^{m}_{R'}.
\end{equation}
Glancing at the calculation above (\ref{eq_charverfinal}), the first term is exactly the vector field $Y = \falpha_{L} = \< \fAd(\falpha), t_{m} \> \cdot t^{m}_{R'}$. Comparing the both expressions now gives the equation
\begin{equation}
\falpha - \fPi^{\ast}(\fa) = \fAd(\falpha) + [\fPi^{\ast}(t_{k}), t^{k}]_{\g^{\ast}}. 
\end{equation}
We thus find $\fb = \falpha - \fAd(\falpha)$. Plugging this back to the expression (\ref{eq_W'zetaRmain2}) gives
\begin{equation}
W'(\zeta_{R'}) = \<\zeta, \frac{1}{2} [ \~\fB(t^{k}), t_{k}]_{\g} + \ad^{\ast}_{\fPi^{\ast}(t_{k})} \~\fB(t^{k}) \>  + \frac{1}{2}\< \zeta, \fa + \~\fB \cdot \fAd(\falpha)\>\\
\end{equation}
This is exactly  what was to be proved, namely (\ref{eq_W'zetaRmain1}). 
\end{proof}
At this moment, it is not clear how the expression (\ref{eq_W'zetaRmain1}) can be useful to calculate the dilaton function $\~\phi \in C^{\infty}(G^{\ast})$. The answer is given in the following proposition.
\begin{tvrz} \label{tvrz_dilatondiffequation}
Consider the smooth function $\~\phi \in C^{\infty}(G^{\ast})$ defined by	 
\begin{equation} \label{eq_phi}
\~\phi := - \frac{1}{2} \ln{|\det( \f1_{\g} + g_{0}^{-1}( \fPi^{\ast} + B_{0}))|}.
\end{equation}
Then its differential can be, for every $\zeta \in \g^{\ast}$, written as 
\begin{equation} \label{eq_dphizetaexpl}
\< d\~\phi, \zeta_{R'} \> = \< \zeta, \frac{1}{2} [\~\fB(t^{k}),t_{k}]_{\g} + \ad^{\ast}_{\fPi^{\ast}(t_{k})} \~\fB(t^{k}) \>.
\end{equation}
In other words, by plugging into (\ref{eq_W'zetaRmain1}), we can rewrite the $1$-form $W'$ as 
\begin{equation}
W'(\zeta_{R'}) = \< d\~\phi, \zeta_{R'} \> + \frac{1}{2} \< \zeta, \fa + \~\fB \cdot \fAd(\falpha) \>.
\end{equation}
\end{tvrz}
\begin{proof}
In fact, the absolute value in (\ref{eq_phi}) is not necessary. Indeed, for positive definite $g_{0}$, the determinant is always positive. However, to keep the validity of all formulas for indefinite scenario, we continue to write the absolute values where it is appropriate.

Now to the actual proof. Write $\~\phi = - \frac{1}{2} \ln{| \det ( \fM)|}$. Then
\begin{equation}
d\~\phi = - \frac{1}{2} {(\fM^{-1})^{j}}_{i} \cdot d( {\fM^{i}}_{j}),
\end{equation}
where ${\fM^{i}}_{j} := \< t^{i}, \fM(t_{j}) \> = \delta^{i}_{j} + g_{0}^{ik} ( \fPi^{\ast}(t_{k},t_{j}) + (B_{0})_{kj})$. It is easy to calculate its exterior derivative using the relation (\ref{eq_useful}). We find the expression
\begin{equation}
\< d( {\fM^{i}}_{j}), \zeta_{R'} \> = g_{0}^{ik} (\zeta_{R'}.\fPi^{\ast}(t_{k},t_{j})) = g_{0}^{ik} \big( \<\zeta, [t_{k},t_{j}]_{\g} \> - \fPi^{\ast}( \ad^{\ast}_{\zeta} t_{k}, t_{j}) - \fPi^{\ast}( t_{k}, \ad^{\ast}_{\zeta} t_{j}) \big).
\end{equation}
Next, note that ${(\fM^{-1})^{j}}_{i} = {( \~\fE \cdot g_{0})^{j}}_{i}$. Using this, we can write $\< d\~\phi, \zeta_{R'}\>$ as 
\begin{equation}
\< d\~\phi, \zeta_{R'} \> = - \frac{1}{2} \big( \<\zeta, [t_{k}, \~\fE(t^{k})]_{\g} \> - \fPi^{\ast}( \ad^{\ast}_{\zeta} t_{k}, \~\fE(t^{k})) - \fPi^{\ast}( t_{k}, \ad^{\ast}_{\zeta} \~\fE(t^{k})) \big).
\end{equation}
Observe that only the skew-symmetric part $\~\fB$ contributes to the sum in the first term. For the other two terms, we shift $\zeta$ to the left, obtaining 
\begin{equation}
\fPi^{\ast}( \ad^{\ast}_{\zeta} t_{k}, \~\fE(t^{k})) + \fPi^{\ast}( t_{k}, \ad^{\ast}_{\zeta} \~\fE(t^{k})) = \< \zeta, \ad^{\ast}_{\fPi^{\ast}(t_{k})} (\~\fE - \~\fE^{T})(t^{k}) \> = 2 \< \zeta, \ad^{\ast}_{\fPi^{\ast}(t_{k})} \~\fB(t^{k}) \>.
\end{equation}
Plugging this back into the above expression yields
\begin{equation}
\< d\~\phi, \zeta_{R'} \> = \< \zeta, \frac{1}{2} [\~\fB(t^{k}),t_{k}]_{\g} + ad^{\ast}_{\fPi^{\ast}(t_{k})} \~\fB(t^{k}) \>.
\end{equation}
But this is exactly the expression (\ref{eq_dphizetaexpl}).
\end{proof}
Now, recall that we have chosen $\cD^{0}$ to be the particular connection from Example \ref{ex_quadrconn}. However, to find the general case is actually quite simple. Therefore, assume that $\cD^{0} \in \LC(\d,\E)$ is an arbitrary Levi-Civita connection on $\d$. Note that the $1$-form $W'$ can be calculated directly from the divergence operator of the connection $\~\cD \in \LC( \mathbb{T}G^{\ast}, V_{+}^{\g})$. This can be easily seen  from Proposition $6.1$ and equation (115) in our paper \cite{Jurco:2016emw}. One has 
\begin{equation} \label{eq_W'asDivergence}
W'(Y) = \frac{1}{2}(\Div_{\~g}(Y) - \Div_{\~\cD}(Y,0)),
\end{equation}
for all $Y \in \X(G^{\ast})$, where $\Div_{\~g}$ denotes the usual covariant divergence operator of the Levi-Civita connection corresponding to the metric $\~g$. One can easily calculate $\Div_{\~\cD}$ using the connection $\cD^{\g}$ and the isomorphism $\fPsi \in \Hom(\gTGa, E')$. One obtains 
\begin{equation}
\begin{split}
\Div_{\~\cD}(Y,0) = & \ \Div_{\~\cD}( \<Y, t_{k}^{R'} \> \cdot (t^{k}_{R'}, 0)) = t^{k}_{R'}.\<Y,t_{k}^{R'}\> + \< Y, t_{k}^{R'}\> \cdot \Div_{\~\cD}(t^{k}_{R'},0) \\
= & \ t^{k}_{R'}.\<Y,t_{k}^{R'}\> + \<Y, t_{k}^{R'} \> \cdot \Div_{\cD^{\g}}( \fPsi(t^{k}_{R'},0)) \\ 
= & \ t^{k}_{R'}.\<Y,t_{k}^{R'} \> + \<Y, t_{k}^{R'} \> \cdot \Div_{\cD^{0}}(0,t^{k}) \\
= & \ \Div_{\~g}(Y) + \<Y, \cDL_{t^{k}_{R'}} t_{k}^{R'} + \ft^{R'} \>,
\end{split}
\end{equation}
where we have used the definition of $\cD^{\g}$ and $\ft \in \g$ is defined by (\ref{eq_DivcD0}). Plugging this expression into (\ref{eq_W'asDivergence}), we see that for all $\zeta \in \g^{\ast}$, one has
\begin{equation} \label{eq_W'zetaRfinal}
W'(\zeta_{R'}) = - \frac{1}{2} \big( \<\zeta,\ft \> + \< \zeta_{R'}, \cDL_{t^{k}_{R'}} t_{k}^{R'} \> \big). 
\end{equation}
We conclude that the $1$-form $W'$ depends only on the divergence of the connection $\cD^{0}$, as the other term is fully determined by the choice of the generalized metric $\E \subseteq \d$ and the geometrical setting. On the other hand, the connection $\cD^{0}$ can be always expressed as a sum
\begin{equation} \label{eq_cD0astcD0}
\cD^{0}_{x}y = \tilde{\cD}_{x}^{0} y + g_{\d}^{-1} \mathbf{k}(x,y,\cdot),
\end{equation}
for all $x,y \in \d$, where $\tilde{\cD}_{x}^{0} \in \LC(\d,\E)$ now denotes the connection (\ref{eq_cD0formula}), and $\mathbf{k} \in \d^{\ast} \otimes \Lambda \d^{\ast}$ is a unique tensor subject to certain properties following from the fact that $\cD^{0} \in \LC(\d,\E)$. One gets
\begin{equation}
\Div_{\cD^{0}}(x) = - \mathbf{k}( t^{\alpha}_{\d}, t_{\alpha}, x),
\end{equation}
for all $x \in \d$. In other words, this is to argue that the term $-\frac{1}{2} \<\zeta, \ft\>$ in (\ref{eq_W'zetaRfinal}) is precisely the one coming from the non-trivial tensor $\mathbf{k}$ in (\ref{eq_cD0astcD0}) measuring the difference of general Levi-Civita connection on $\d$ with respect to the generalized metric $\E$ from the special one defined by (\ref{eq_cD0formula}). This observation allows us to formulate the main theorem of this subsection. 

\begin{theorem}[\textbf{Dilaton in Poisson--Lie T-duality}] \label{thm_dilaton}
Let $\cD^{0} \in \LC(\d,\E)$ be an arbitrary Levi-Civita connection on $\d$ with respect to the generalized metric $\E$. Let $\~\cD \in \LC(\gTGa, V_{+}^{\g})$ be the corresponding connection on $\gTGa$ defined via (\ref{eq_cDastdef}). Let $W' \in \Omega^{1}(G^{\ast})$ be the $1$-form whose value on the vector field is given by the left-hand side of (\ref{eq_dilatonrel}), and define $\~\phi \in C^{\infty}(G^{\ast})$ as 
\begin{equation} \label{eq_phiastfinalexpl}
\~\phi = - \frac{1}{2} \ln | \det( \f1_{\g} + g_{0}^{-1}( \fPi^{\ast} + B_{0})) |.
\end{equation}
Then, for all $\zeta \in \g^{\ast}$, the value of  $W'$ on the right-invariant vector field $\zeta_{R'}$ has the form 
\begin{equation}
W'(\zeta_{R'}) = \< d\~\phi, \zeta_{R'} \> + \frac{1}{2} \< \zeta, \fa - \ft + \~\fB \cdot \fAd(\falpha) \>, 
\end{equation}
where $\< (\ft,\ftau), (x,\xi) \>_{\d} = \Div_{\cD^{0}}(x,\xi)$ and $\< (\fa,\falpha), (x,\xi) \>_{\d} = \Tr(ad_{x}) + \Tr(ad^{\ast}_{\xi})$, for all $(x,\xi) \in \d$. 
\end{theorem}
It is not easy to obtain a general answer to the question of exactness of $W'$. Note that both $\fa$ and $\ft$ have to be $\Ad^{\ast}$-invariant elements of $\g$. We have to assume so for $\ft$, see Theorem \ref{thm_charverfield}. For $\fa$, this can be shown as follows. Let $\xi,\eta \in \g^{\ast}$. Then
\begin{equation}
(\ad^{\ast}_{\xi} \fa)(\eta) = \< \fa, [\eta,\xi]_{\g^{\ast}} \> = \Tr( \ad_{[\eta,\xi]_{\g^{\ast}}}) = \Tr( [\ad_{\eta}, \ad_{\xi}] ) = 0,
\end{equation}
as the commutator of two linear operators is always traceless. This observation implies that the $1$-form $( \fa - \ft)^{R'} = ( \fa - \ft)^{L'}$ is bi-invariant and thus closed. However, the term $\~\fB \cdot \fAd(\falpha)$ corresponds to the $1$-form $\~B( \falpha_{L'})$ which is \emph{not closed} in general. It is thus a safe bet to assume $(\fa,\falpha) = (0,0)$. In this case, we can choose $\cD^{0}$ so that $(\ft,\ftau) = (0,0)$, in particular we can work with the one in Example \ref{ex_quadrconn}. We thus introduce the following notion:
\begin{definice} Let $(\g,\delta)$ be a Lie bialgebra. We say that $(\g,\delta)$ is a \textbf{unimodular Lie bialgebra} if $(\g,[\cdot,\cdot]_{\g})$ is a unimodular Lie algebra and $\delta$ satisfies the condition
\begin{equation}
\delta(t_{k}) (\eta \otimes t^{k}) = 0, 
\end{equation}
for all $\eta \in \g^{\ast}$. Equivalently, the corresponding dual Lie algebra $[\cdot,\cdot]_{\g^{\ast}} = \delta^{T}$ is also unimodular. 
\end{definice}

We will henceforth \textbf{assume that $(\g,\delta)$ is a unimodular Lie bialgebra}. It follows from Example \ref{ex_unimod} and Theorem \ref{thm_dilaton} that we can choose $\cD^{0}$ to be the connection (\ref{eq_cD0formula}) and the corresponding connections $\cD \in \LC(\gTG, V_{+}^{\g^{\ast}})$ and $\~\cD \in \LC( \gTGa, V_{+}^{\g})$ then both satisfy the assumptions of Theorem \ref{thm_eomconn}. Moreover, we have the explicit formula (\ref{eq_phiastfinalexpl}) for both dilaton fields.
\subsection{Analysis of the dilaton} \label{subsec_analysis}
For the sake of completeness, we can now analyze the case when only one of the Lie algebras is unimodular. We still assume that $(\ft,\ftau) = (0,0)$ and $\falpha = 0$. However, we allow $\fa \neq 0$. We have
\begin{equation}
W'(\zeta_{R'}) = \< d \~\phi, \zeta_{R'} \> + \frac{1}{2} \<\zeta,\fa\>. 
\end{equation}
It turns out that the biinvariant form $\fa^{R'} \in \Omega^{1}(G^{\ast})$ is in fact exact, namely $\fa^{R'} = \ln(\det(\fAd))$. This can be proved by a direct calculation. One thus finds $W'(\zeta_{R'}) = \< d \~\phi', \zeta_{R'} \>$ for
\begin{equation} \label{eq_phi'dilaton}
\~\phi' = \~\phi_{0} - \frac{1}{2} \ln | \det(\f1_{\g} + g_{0}^{-1}( \fPi^{\ast} + B_{0})) | + \frac{1}{2} \ln (\det(\fAd)),
\end{equation}
where $\~\phi_{0} \in \R$ in some additive constant. After some simple manipulations, we can write
\begin{equation}
\~\phi' = - \frac{1}{2} \{ \~\phi'_{0} + \ln|\det(E_{0}^{-1})| - \ln| \det(\~\fE)| - \ln(\det(\fAd)) \},
\end{equation}
for some constant $\~\phi'_{0} \in \R$. One can now compare this formula\footnote{More precisely its dual version with the role of $G$ and $G^{\ast}$ interchanged.} with (3.13) in \cite{VonUnge:2002xjf}. The only difference is the overall prefactor $-1/2$, presumably coming from a different parametrization of the effective action. Note that in the earlier literature \cite{Tyurin:1995bu, Bossard:2001au}, the term containing $\det(\fAd)$ is missing. However, this does not matter if $(\g,\delta)$ is a unimodular Lie bialgebra. 

Moreover, as pointed to us by the referee, on can also rewrite $\~\phi'$ as follows. As $\d$ is a quadratic Lie algebra, it is always unimodular. In this subsection, we assume that $\g$ is unimodular. When both Lie algebras of the Manin pair $(\d,\g)$ are unimodular, there is a canonical volume form on the quotient $D/G$ determined uniquely up to a non-zero multiplicative constant. Under the diffeomorphism $D/G \cong G^{\ast}$, this volume form can be chosen to be the \textit{left} Haar measure $\mu \in \Omega^{n}(G^{\ast})$. Assume that we have fixed a basis $(t^{k})_{k=1}^{n}$ of $\g^{\ast}$ and set $\mu$ to be
\begin{equation} \label{eq_muform}
\mu = t_{1}^{L'} \^ \cdots \^ t_{n}^{L'} = \det(\fAd)^{-1} \cdot t_{1}^{R'} \^ \cdots \^ t_{n}^{R'}.
\end{equation}
Next, the comparison of symmetric parts of (\ref{eq_EastasE0andPi}) yields the formula $(\~\fg - \~\fB \~\fg^{-1} \~\fB) = g_{0}^{-1}$. Moreover, for any invertible matrix $\~\fE = \~\fg + \~\fB$ with the invertible symmetric part $\~\fg$, one has
\begin{equation}
| \det(\~\fE)| = | \det(\~\fg) |^{\frac{1}{2}} \cdot | \det( \~\fg - \~\fB \~\fg^{-1} \~\fB) |^{\frac{1}{2}}
\end{equation}
Using these two observations, it straightforward to rewrite (\ref{eq_phi'dilaton}) as 
\begin{equation}
\~\phi' = \~\phi_{0} + \frac{1}{4} \ln|\det(g_{0})| + \frac{1}{2} \ln( |\det(\~\fg)|^{\frac{1}{2}} \cdot \det(\fAd) )
\end{equation}
If $(t^{k})_{k=1}^{n}$ is a right-handed basis of $\g^{\ast}$, we can write the volume form $d\vol_{\~g}$ as 
\begin{equation}
d\vol_{\~g} = | \det(\~\fg)|^{\frac{1}{2}} \cdot t_{1}^{R'} \^ \cdots \^ t_{n}^{R'}
\end{equation}
Finally, comparing this to (\ref{eq_phi'dilaton}), we may rewrite the dilaton field as
\begin{equation} \label{eq_dilatonsasvolumeforms}
\~\phi' = \~\phi''_{0} + \frac{1}{2} \ln| \frac{d\vol_{\~g}}{\mu}|,
\end{equation}
where $\~\phi''_{0} \in \R$ is an irrelevant constant. Note that $\mu$ can be any left Haar measure, the resulting difference is then absorbed into $\~\phi''_{0}$. 
\section{Poisson--Lie T-duality of effective actions} \label{sec_PLT}
Finally, we may combine the results of all the above sections to derive the main statement of this paper. In the previous section, we have constructed a pair of connections $\cD \in \LC(\gTG, V_{+}^{\g^{\ast}})$ and $\~\cD \in \LC( \gTGa, V_{+}^{\g})$ which both satisfy the assumptions of Theorem \ref{thm_eomconn}. Now, we can examine their mutual relation in terms of the properties of their respective curvatures. 
\begin{tvrz} \label{tvrz_conpropeq}
Let $\cD^{0} \in \LC(\d,\E)$ be an arbitrary Levi-Civita connection on $\d$ with respect to the generalized metric $\E \subseteq \d$. Let $\~\cD \in \LC(\gTGa, V_{+}^{\g})$ be the Levi-Civita connection on $\gTGa$ defined by (\ref{eq_cDastdef}) using the connection $\cD^{\g} \in \LC(E',\E'_{\g})$ obtained in (\ref{eq_cDgfinal}) of Subsection \ref{subsec_conred}. Let $\gm_{0}$ be the positive-definite metric on $\d$ corresponding to $\E \subseteq \d$ and $\~\gm$ be the positive-definite fiber-wise metric on $\gTGa$ corresponding to $V^{\g}_{+} \subseteq \gTGa$. 

Then the scalar curvature $\RS_{\~\cD}^{\~\gm} \in C^{\infty}(G^{\ast})$ is constant and equal to $\RS_{\cD^{0}}^{\gm_{0}}$. Moreover, the connection $\~\cD$ is Ricci compatible with the generalized metric $V_{+}^{\g}$ if and only if the connection $\cD^{0}$ is Ricci compatible with the generalized metric $\E$. 
\end{tvrz}
\begin{proof}
The statement is a consequence of the fact that scalar curvatures and the Ricci compatibility are preserved by Courant algebroid isomorphisms. Observe that both $\~\cD$ and $V_{+}^{\g}$ are constructed from $\cD^{\g}$ and $\E'_{\g}$ using the Courant algebroid isomorphism $\fPsi \in \Hom(\gTGa,E')$. Moreover, the curvature tensor, the Ricci tensor and the corresponding scalar curvatures of $\cD^{\g}$ can be calculated using only the constant sections of $E'$. But the formula (\ref{eq_cDgfinal}) then implies that all those quantities are constant on $G^{\ast}$ and equal to those of $\cD^{0}$. Similarly, as $\E'_{\g} = G^{\ast} \times \E$, the same stands true for the Ricci compatibility. 
\end{proof}

Obviously, completely the same statement holds for the relation between the connection $\cD^{0}$ and the Levi-Civita connection $\cD \in \LC( \gTG, V_{+}^{\g^{\ast}})$. This quite simple observation immediately implies the main theorem of this work. We will include the notation to make it self-contained. 
\begin{theorem}[\textbf{Poisson--Lie T-duality of effective actions}] \label{thm_themain}
Let $(\g,\delta)$ be a unimodular Lie bialgebra. Let $(G,\Pi)$ and $(G^{\ast},\Pi^{\ast})$ be two mutually dual connected Poisson--Lie groups integrating the Lie bialgebra $(\g,\delta)$ and the dual one $(\g^{\ast},\delta^{\ast})$, respectively. 

Let $\E \subseteq \d$ be a generalized metric on $\d$. One can write $\E = \{ (x,E_{0}(x)) \; | \; x \in \g \}$, for a linear isomorphism $E_{0} \in \Hom(\g,\g^{\ast})$ such that $E_{0} = g_{0} + B_{0}$ for a positive definite metric $g_{0} \in S^{2} \g^{\ast}$ and a $2$-form $B_{0} \in \Lambda^{2}\g^{\ast}$. Let $\fPi \in C^{\infty}(G, \Lambda^{2}\g)$ and $\fPi^{\ast} \in C^{\infty}(G^{\ast}, \Lambda^{2} \g^{\ast})$ be defined as 
\begin{equation}
\fPi(\xi,\eta) = \Pi( \xi_{R}, \eta_{R}), \; \; \fPi^{\ast}(x,y) = \Pi^{\ast}(x^{R'},y^{R'}),
\end{equation}
for all $\xi,\eta \in \g^{\ast}$ and $x,y \in \g$. Define $\fE \in C^{\infty}(G,\Hom(\g,\g^{\ast}))$ and $\~\fE \in C^{\infty}(G^{\ast}, \Hom(\g^{\ast},\g))$ by
\begin{equation}
\fE = ( E_{0}^{-1} + \fPi)^{-1}, \; \; \~\fE = (E_{0} + \fPi^{\ast})^{-1}.
\end{equation} 
Decompose these maps into the symmetric and skew-symmetric parts, that is, $\fE = \fg + \fB$ and $\~\fE = \~\fg + \~\fB$, and define the fields $(g,B)$ and $(\~g,\~B)$ by 
\begin{align}
g(x^{R},y^{R}) = \fg(x,y), \; \; B(x^{R},y^{R}) = \fB(x,y), \\
\~g(\xi_{R'},\eta_{R'}) = \~\fg(\xi,\eta), \; \; \~B(\xi_{R'},\eta_{R'}) = \~\fB(\xi,\eta),
\end{align}
for all $x,y \in \g$ and $\xi,\eta \in \g^{\ast}$. 

Next, let $\phi \in C^{\infty}(G)$ and $\~\phi \in C^{\infty}(G^{\ast})$ be the scalar functions defined\footnote{According to Subsection \ref{subsec_analysis}, we can rewrite the dilaton fields using the volume forms as in (\ref{eq_dilatonsasvolumeforms}). However, we have decided to keep them in this form as it is more calculation friendly.}as
\begin{align}
\label{eq_dilaton1} \phi & = - \frac{1}{2} \ln | \det( \f1_{\g^{\ast}} + \~g^{-1}_{0} ( \fPi + \~B_{0})) |, \\
\label{eq_dilaton2} \~\phi &= - \frac{1}{2} \ln | \det( \f1_{\g} + g_{0}^{-1}(\fPi^{\ast} + B_{0})) |, 
\end{align}
where $(\~g_{0},\~B_{0})$ are the fields dual to $(g_{0},B_{0})$, that is $E_{0}^{-1} = \~g_{0} + \~B_{0}$.  Finally, define the actions
\begin{align}
\label{eq_akce1} S_{G}[g,B,\phi] & = \int_{G} e^{-2\phi} \{ \RS(g) - \frac{1}{2} \<dB,dB\>_{g} + 4 \| \cD^{g} \phi \|^{2}_{g} \} \cdot d \vol_{g}, \\
\label{eq_akce2} S_{G^{\ast}}[\~g,\~B,\~\phi] & = \int_{G^{\ast}} e^{-2\~\phi} \{ \RS(\~g) - \frac{1}{2} \<d\~B,d\~B\>_{\~g} + 4 \| \cD^{\~g} \~\phi \|^{2}_{\~g} \} \cdot d \vol_{\~g},
\end{align}
where $\RS(g)$ and $\RS(\~g)$ are the usual scalar curvatures of the metric $g$ and $\~g$, respectively. 

Then the fields $(g,B,\phi)$ satisfy the equations of motion given by the action $S_{G}$ if and only if $(\~g,\~B,\~\phi)$ satisfy the equations of motion given by the action $S_{G^{\ast}}$. We call this the \textbf{Poisson--Lie T-duality of low-energy string effective actions}. 
\end{theorem}
\begin{proof}
The proof is a straightforward application of Theorem \ref{thm_eomconn}. First, define $\cD^{0} \in \LC(\d,\E)$ to be the connection (\ref{eq_cD0formula}). Using the procedure described in Section \ref{sec_dilatons}, we can construct a pair of Levi-Civita connections $\cD \in \LC(\gTG, V_{+}^{\g^{\ast}})$ and $\~\cD \in \LC(\gTGa, V_{+}^{\g})$. The generalized metric $V_{+}^{\g^{\ast}} \subseteq \gTG$ corresponds to the pair $(g,B)$ from the assumptions of this theorem, and similarly for $V_{+}^{\g} \subseteq \gTGa$ and $(\~g,\~B)$. As $(\g,\delta)$ is assumed to be unimodular, it follows from Theorem \ref{thm_dilaton} and the discussion following it that both $\cD$ and $\~\cD$ satisfy the assumptions of Theorem \ref{thm_eomconn} for $\phi$ and $\~\phi$ given by (\ref{eq_dilaton1}) and (\ref{eq_dilaton2}), respectively. 

According to Theorem \ref{thm_eomconn}, $(g,B,\phi)$ satisfy the equations of motion given by (\ref{eq_akce1}) if and only if $\RS_{\cD}^{\gm} = 0$ and $\cD$ is Ricci compatible with the generalized metric $V^{\g^{\ast}}_{+} \subseteq \gTG$. Here $\gm$ is the positive-definite fiber-wise metric (\ref{eq_genmetricusual}) corresponding to $V_{+}^{\g^{\ast}}$. By Proposition \ref{tvrz_conpropeq}, this is equivalent to $\RS_{\cD^{0}}^{\gm_{0}} = 0$ and the Ricci compatibility of $\cD^{0}$ with $\E$. But the same is true for the dual fields $(\~g,\~B,\~\phi)$. We conclude that both systems of equations of motion are equivalent to the same set of algebraic equations.
\end{proof}
\begin{rem}
Some remarks are in order.
\begin{enumerate}[(i)]
\item In the proof, we have chosen $\cD^{0}$ to be the one defined by (\ref{eq_cD0formula}). Assuming that $(\g,\delta)$ is unimodular, we only have to choose $\cD^{0}$ with $(\ft,\ftau) = (0,0)$. Similarly to the discussion above Theorem \ref{thm_dilaton}, one can show that the scalar curvature $\RS_{\cD^{0}}^{\gm_{0}}$ and the Ricci compatibility of $\cD^{0}$ with $\E$ depend only on the divergence operator and thus only on $(\ft, \ftau)$. 

\item The interesting aspect of the Poisson--Lie T-duality is that we start with a piece of algebraic data (the generalized metric $\E \subseteq \d$) and use it to obtain the geometric data $(g,B,\phi)$ and $(\~g,\~B,\~\phi)$ on the mutually dual Poisson--Lie groups $(G,\Pi)$ and $(G^{\ast},\Pi^{\ast})$, respectively. It should not be surprising that the equations of motion of the corresponding low-energy string effective actions are in fact equivalent to the system of algebraic equations for the fields $(g_{0},B_{0})$ determining the linear subspace $\E \subseteq \d$. It would be interesting to find some natural interpretation of these equations. 

\item The whole theory works to some extent  also for indefinite metrics. If the symmetric part $g_{0}$ of the map $E_{0}$ is not positive-definite, problems on various places can occur. In particular, some of the maps involved need not to be invertible, e.g. the map $E_{0} + \Pi^{\ast}$,  whose inverse $\~\fE$ then provides for the backgrounds $(\~g,\~B)$. However, everything can still work if one is careful enough to avoid these invertibility problems. 

\item Note that most of Lie algebras interesting for applications in physics are unimodular. In particular, every quadratic Lie algebra is unimodular. This includes e.g. all semisimple Lie algebras. 

However, there is a lot of Lie bialgebras which are not unimodular. For example, in the classification of six-dimensional Drinfeld doubles in \cite{Snobl:2002kq}, any Manin triple where one of the Lie algebras is isomorphic to the one of the Bianchi algebras in the set $\{ \mathbf{3, 4, 5, 6_{a}, 7_{a}} \}$ does not correspond to an unimodular Lie bialgebra. 

\item In the proof of Proposition \ref{tvrz_dilatondiffequation}, we have shown that the dilaton field $(\ref{eq_phi})$ is a smooth function solving the equation $d\~\phi = W'$. It follows that it is determined only up to an additive constant. However, its choice is completely irrelevant for the low-energy effective string actions, as it only multiplies the corresponding functional by a non-zero real constant. 

\item Note that all arguments used in this paper work also for the more general consideration of two different Manin triples $(\d,\g,\g')$ and $(\d,\~\g,\~\g')$ with the same quadratic Lie algebra $(\d,\<\cdot,\cdot\>_{\d})$, assuming that $\g$ and $\~\g$ are unimodular, respectively. We can then compare the effective actions on Lie groups $G^{\ast}$ and $\~G^{\ast}$, provided the formulas for dilaton fields are generalized as in (\ref{eq_phi'dilaton}). 

Observe that the connection $\cD^{0}$ defined by (\ref{eq_cD0formula}) is independent of any choices of the subspaces $\g$ and $\~\g$ or their complements. We can thus immediately generalize Theorem \ref{thm_themain} to include the so called \textbf{Poisson--Lie T-plurality}, see \cite{VonUnge:2002xjf}. 

\item So far, we have considered only the simplest case where the split Manin pair $(\d,\g,j)$ corresponds to a Manin triple $(\d,\g,\g^{\ast})$ and thus to a Lie bialgebra $(\g,\delta)$. However, it should be possible to generalize the results in this paper to the case where $(\d,\g,j)$ gives rise to a \textbf{Manin quasi-triple} $(\d,\g,\g')$ and consequently to a \textbf{Lie quasi-bialgebra} $(\g,\delta,\mu)$. One would have to employ the techniques involving quasi-Poisson structures on Lie groups, see \cite{alekseev2000manin}. More generally, one can consider a general principal bundle with the structure Lie group $D$. We plan to address these two generalizations in the near future. 
\end{enumerate}
\end{rem}
\section*{Acknowledgments}
We would like to thank Ladislav Hlavatý for the crucial hint leading to the correct formula for a dilaton field, and Mario Garcia-Fernandez for helpful discussions. We also thank the anonymous referee for useful comments leading to the discussion in Subsection \ref{subsec_analysis}. 

The research of B.J. was supported by grant GAČR P201/12/G028, he would like to thank the Max Planck Institute for Mathematics for hospitality.

The research of J.V. was supported by RVO: 67985840, and he would like to thank the Max Planck Institute for Mathematics in for hospitality. The author would like to acknowledge the contribution of the COST Action MP1405.
\newpage
\bibliography{bib} 

\providecommand{\href}[2]{#2}\begingroup\raggedright\begin{thebibliography}{10}

\bibitem{1995PhLB..351..455K}
C.~{Klim{\v{c}}{\'{\i}}k} and P.~{{\v{S}}evera}, \emph{{Dual non-Abelian
  duality and the Drinfeld double}},
  \href{https://doi.org/10.1016/0370-2693(95)00451-P}{\emph{Physics Letters B}
  {\bfseries 351} (Feb., 1995) 455--462},
  [\href{https://arxiv.org/abs/hep-th/9502122}{{\ttfamily hep-th/9502122}}].

\bibitem{Klimcik:1995jn}
C.~Klim\v{c}\'{i}k, \emph{{Poisson-Lie T duality}},
  \href{https://doi.org/10.1016/0920-5632(96)00013-8}{\emph{Nucl. Phys. Proc.
  Suppl.} {\bfseries 46} (1996) 116--121},
  [\href{https://arxiv.org/abs/hep-th/9509095}{{\ttfamily hep-th/9509095}}].

\bibitem{Klimcik:1995dy}
C.~Klim\v{c}\'{i}k and P.~\v{S}evera, \emph{{Poisson-Lie T duality and loop
  groups of Drinfeld doubles}},
  \href{https://doi.org/10.1016/0370-2693(96)00025-1}{\emph{Phys.Lett.}
  {\bfseries B372} (1996) 65--71},
  [\href{https://arxiv.org/abs/hep-th/9512040}{{\ttfamily hep-th/9512040}}].

\bibitem{Alekseev:1995ym}
A.~Alekseev, C.~Klimcik and A.~A. Tseytlin, \emph{{Quantum Poisson-Lie T
  duality and WZNW model}},
  \href{https://doi.org/10.1016/0550-3213(95)00575-7}{\emph{Nucl. Phys.}
  {\bfseries B458} (1996) 430--444},
  [\href{https://arxiv.org/abs/hep-th/9509123}{{\ttfamily hep-th/9509123}}].

\bibitem{Tyurin:1995bu}
E.~Tyurin and R.~von Unge, \emph{{Poisson-lie T duality: The Path integral
  derivation}}, \href{https://doi.org/10.1016/0370-2693(96)00680-6}{\emph{Phys.
  Lett.} {\bfseries B382} (1996) 233--240},
  [\href{https://arxiv.org/abs/hep-th/9512025}{{\ttfamily hep-th/9512025}}].

\bibitem{Hassler:2017yza}
F.~Hassler, \emph{{Poisson-Lie T-Duality in Double Field Theory}},
  \href{https://arxiv.org/abs/1707.08624}{{\ttfamily 1707.08624}}.

\bibitem{Kosmann-Schwarzbach1997}
Y.~Kosmann-Schwarzbach, \emph{Lie bialgebras, poisson Lie groups and dressing
  transformations}, pp.~104--170.
\newblock Springer Berlin Heidelberg, Berlin, Heidelberg, 1997.
\newblock 10.1007/BFb0113695.

\bibitem{polchinski1998string}
J.~Polchinski, \emph{String theory, vol. 1, 2}, {\emph{Cambridge, UK: Univ. Pr}
  (1998) }.

\bibitem{Vysoky:2017epf}
J.~Vysoky, \emph{{Kaluza-Klein Reduction of Low-Energy Effective Actions:
  Geometrical Approach}},
  \href{https://doi.org/10.1007/JHEP08(2017)143}{\emph{JHEP} {\bfseries 08}
  (2017) 143}, [\href{https://arxiv.org/abs/1704.01123}{{\ttfamily
  1704.01123}}].

\bibitem{VonUnge:2002xjf}
R.~Von~Unge, \emph{{Poisson Lie T plurality}},
  \href{https://doi.org/10.1088/1126-6708/2002/07/014}{\emph{JHEP} {\bfseries
  07} (2002) 014}, [\href{https://arxiv.org/abs/hep-th/0205245}{{\ttfamily
  hep-th/0205245}}].

\bibitem{Hlavaty:2006hu}
L.~Hlavat\'y, \emph{{Dilatons in curved backgrounds by the Poisson-Lie
  transformation}},  \href{https://arxiv.org/abs/hep-th/0601172}{{\ttfamily
  hep-th/0601172}}.

\bibitem{Hlavaty:2004mr}
L.~Hlavat\'y and L.~\v{S}nobl, \emph{{Poisson-Lie T-plurality of
  three-dimensional conformally invariant sigma models. II. Nondiagonal metrics
  and dilaton puzzle}},
  \href{https://doi.org/10.1088/1126-6708/2004/10/045}{\emph{JHEP} {\bfseries
  10} (2004) 045}, [\href{https://arxiv.org/abs/hep-th/0408126}{{\ttfamily
  hep-th/0408126}}].

\bibitem{2015LMaPh.tmp...53S}
P.~\v{S}evera, \emph{{Poisson-Lie T-Duality and Courant Algebroids}},
  \href{https://doi.org/10.1007/s11005-015-0796-4}{\emph{Letters in
  Mathematical Physics} (Aug., 2015) },
  [\href{https://arxiv.org/abs/1502.04517}{{\ttfamily 1502.04517}}].

\bibitem{Severa:2016prq}
P.~\v{S}evera, \emph{{Poisson-Lie T-duality as a boundary phenomenon of
  Chern-Simons theory}},
  \href{https://doi.org/10.1007/JHEP05(2016)044}{\emph{JHEP} {\bfseries 05}
  (2016) 044}, [\href{https://arxiv.org/abs/1602.05126}{{\ttfamily
  1602.05126}}].

\bibitem{Severa:2016lwc}
P.~\v{S}evera and F.~Valach, \emph{{Ricci flow, Courant algebroids, and
  renormalization of Poisson--Lie T-duality}},
  \href{https://arxiv.org/abs/1610.09004}{{\ttfamily 1610.09004}}.

\bibitem{2013arXiv1304.4294G}
M.~Garcia-Fernandez, \emph{{Torsion-Free Generalized Connections and Heterotic
  Supergravity}},
  \href{https://doi.org/{10.1007/s00220-014-2143-5}}{\emph{{Commun. Math.
  Phys.}} {\bfseries {332}} ({Nov}, {2014}) {89--115}},
  [\href{https://arxiv.org/abs/1304.4294}{{\ttfamily 1304.4294}}].

\bibitem{Garcia-Fernandez:2016azr}
M.~Garcia-Fernandez, \emph{{Lectures on the Strominger system}},
  \href{https://arxiv.org/abs/1609.02615}{{\ttfamily 1609.02615}}.

\bibitem{Garcia-Fernandez:2016ofz}
M.~Garcia-Fernandez, \emph{{Ricci flow, Killing spinors, and T-duality in
  generalized geometry}},  \href{https://arxiv.org/abs/1611.08926}{{\ttfamily
  1611.08926}}.

\bibitem{Jurco:2015xra}
B.~Jur\v{c}o and J.~Vysok\'y, \emph{{Leibniz algebroids, generalized Bismut
  connections and Einstein-Hilbert actions}},
  \href{https://doi.org/10.1016/j.geomphys.2015.06.017}{\emph{J. Geom. Phys.}
  {\bfseries 97} (2015) 25--33},
  [\href{https://arxiv.org/abs/1503.03069}{{\ttfamily 1503.03069}}].

\bibitem{Jurco:2016emw}
B.~Jur\v{c}o and J.~Vysok\'y, \emph{{Courant Algebroid Connections and String
  Effective Actions}},  in \emph{{Tohoku Forum for Creativity}}, 2016,
  \href{https://arxiv.org/abs/1612.01540}{{\ttfamily 1612.01540}},
  \href{https://inspirehep.net/record/1501961/files/arXiv:1612.01540.pdf}{https://inspirehep.net/record/1501961/files/arXiv:1612.01540.pdf}.

\bibitem{Jurco:2015bfs}
B.~Jur\v{c}o and J.~Vysok\'y, \emph{{Heterotic reduction of Courant algebroid
  connections and Einstein-Hilbert actions}},
  \href{https://doi.org/10.1016/j.nuclphysb.2016.04.038}{\emph{Nucl. Phys.}
  {\bfseries B909} (2016) 86--121},
  [\href{https://arxiv.org/abs/1512.08522}{{\ttfamily 1512.08522}}].

\bibitem{2007arXiv0710.0639B}
H.~{Bursztyn} and M.~{Crainic}, \emph{{Dirac geometry, quasi-Poisson actions
  and D/G-valued moment maps}}, {\emph{ArXiv e-prints} (Oct., 2007) },
  [\href{https://arxiv.org/abs/0710.0639}{{\ttfamily 0710.0639}}].

\bibitem{Mackenzie}
K.~C. Mackenzie, \emph{{General theory of Lie groupoids and Lie algebroids}},
  vol.~213 of \emph{London Mathematical Society Lecture Note Series}.
\newblock Cambridge University Press, Cambridge, 2005.

\bibitem{liu1997manin}
Z.-J. Liu, A.~Weinstein, P.~Xu et~al., \emph{{Manin triples for Lie
  bialgebroids}}, {\emph{J. Differential Geom} {\bfseries 45} (1997) 547--574}.

\bibitem{1999math.....10078R}
D.~{Roytenberg}, \emph{{Courant algebroids, derived brackets and even
  symplectic supermanifolds}}, {\emph{ArXiv Mathematics e-prints} (Oct., 1999)
  }, [\href{https://arxiv.org/abs/math/9910078}{{\ttfamily math/9910078}}].

\bibitem{alekseevxu}
A.~Alekseev and P.~Xu, ``{Derived Brackets and Courant Algebroids}.''.

\bibitem{Severa:2017oew}
P.~\v{S}evera, \emph{{Letters to Alan Weinstein about Courant algebroids}},
  \href{https://arxiv.org/abs/1707.00265}{{\ttfamily 1707.00265}}.

\bibitem{lu1990}
J.-H. Lu and A.~Weinstein, \emph{Poisson lie groups, dressing transformations,
  and bruhat decompositions}, {\emph{J. Differential Geom.} {\bfseries 31}
  (1990) 501--526}.

\bibitem{luthesis}
J.-H. Lu, \emph{Multiplicative and Affine Poisson Structures on Lie Groups},
  Ph.D. thesis.

\bibitem{2007arXiv0710.2719G}
M.~{Gualtieri}, \emph{{Branes on Poisson varieties}}, {\emph{ArXiv e-prints}
  (Oct., 2007) }, [\href{https://arxiv.org/abs/0710.2719}{{\ttfamily
  0710.2719}}].

\bibitem{Hohm:2012mf}
O.~Hohm and B.~Zwiebach, \emph{{Towards an invariant geometry of double field
  theory}}, \href{https://doi.org/10.1063/1.4795513}{\emph{J. Math. Phys.}
  {\bfseries 54} (2013) 032303},
  [\href{https://arxiv.org/abs/1212.1736}{{\ttfamily 1212.1736}}].

\bibitem{Bursztyn2007726}
H.~Bursztyn, G.~R. Cavalcanti and M.~Gualtieri, \emph{{Reduction of Courant
  algebroids and generalized complex structures}}, {\emph{{Advances in
  Mathematics}} {\bfseries {211}} ({2007}) {726--765}},
  [\href{https://arxiv.org/abs/math/0509640}{{\ttfamily math/0509640}}].

\bibitem{Baraglia:2013wua}
D.~Baraglia and P.~Hekmati, \emph{{Transitive Courant Algebroids, String
  Structures and T-duality}},
  \href{https://doi.org/10.4310/ATMP.2015.v19.n3.a3}{\emph{Adv. Theor. Math.
  Phys.} {\bfseries 19} (2015) 613--672},
  [\href{https://arxiv.org/abs/1308.5159}{{\ttfamily 1308.5159}}].

\bibitem{Klimcik:1995ux}
C.~Klim\v{c}\'{i}k and P.~\v{S}evera, \emph{{Dual non-Abelian duality and the
  Drinfeld double}},
  \href{https://doi.org/10.1016/0370-2693(95)00451-P}{\emph{Phys. Lett.}
  {\bfseries B351} (1995) 455--462},
  [\href{https://arxiv.org/abs/hep-th/9502122}{{\ttfamily hep-th/9502122}}].

\bibitem{Bossard:2001au}
A.~Bossard and N.~Mohammedi, \emph{{Poisson-Lie duality in the string effective
  action}}, \href{https://doi.org/10.1016/S0550-3213(01)00541-7}{\emph{Nucl.
  Phys.} {\bfseries B619} (2001) 128--154},
  [\href{https://arxiv.org/abs/hep-th/0106211}{{\ttfamily hep-th/0106211}}].

\bibitem{Snobl:2002kq}
L.~\v{S}nobl and L.~Hlavat\'y, \emph{{Classification of six-dimensional real
  Drinfeld doubles}},
  \href{https://doi.org/10.1142/S0217751X02010571}{\emph{Int. J. Mod. Phys.}
  {\bfseries A17} (2002) 4043--4068},
  [\href{https://arxiv.org/abs/math/0202210}{{\ttfamily math/0202210}}].

\bibitem{alekseev2000manin}
A.~Alekseev, Y.~Kosmann-Schwarzbach et~al., \emph{Manin pairs and moment maps},
  {\emph{Journal of Differential geometry} {\bfseries 56} (2000) 133--165}.

\end{thebibliography}\endgroup
\end{document}